\documentclass[journal,twoside,web]{ieeecolor}
\usepackage{generic}
\usepackage{cite}
\usepackage{amsmath,amssymb,amsfonts}
\usepackage{algorithmic}
\usepackage{graphicx}
\usepackage{epstopdf}
\usepackage{algorithm,algorithmic}
\usepackage{hyperref}
\usepackage{textcomp}
\usepackage{setspace}
\usepackage{subcaption}
\usepackage{ifthen}
\hypersetup{hidelinks=true}

\newtheorem{lemma}{\textbf{Lemma}}
\newtheorem{remark}{\textbf{Remark}}
\newtheorem{theorem}{\textbf{Theorem}}

\newtheorem{assumption}{\textbf{Assumption}}

\newtheorem{definition}{\textbf{Definition}}

\newtheorem{example}{\textbf{Example}}
\newcommand{\norm}[1]{\left\Vert #1 \right\Vert}

\newboolean{showtodofuture}
\setboolean{showtodofuture}{false}
\newboolean{showtodo}
\setboolean{showtodo}{true}
\newboolean{shownotetomyself}
\setboolean{shownotetomyself}{false}

\newcommand{\todofuture}[1]{\ifthenelse{\boolean{showtodofuture}}{\textcolor{blue}{(TODO: #1)}}{}}
\newcommand{\todo}[1]{\ifthenelse{\boolean{showtodo}}{\textcolor{red}{(TODO: #1)}}{}}
\newcommand{\notetomyself}[1]{\ifthenelse{\boolean{shownotetomyself}}{\textcolor{red}{(Note: #1)}}{}}

\newcommand{\intoninf}{\Theta}

\def\BibTeX{{\rm B\kern-.05em{\sc i\kern-.025em b}\kern-.08em
    T\kern-.1667em\lower.7ex\hbox{E}\kern-.125emX}}
\allowdisplaybreaks

\begin{document}
\title{
    Optimal Control in Both Steady State and Transient Process with Unknown Disturbances
}


\author{Ming Li, Zhaojian Wang, Feng Liu, Ming Cao, \IEEEmembership{Fellow, IEEE}, and Bo Yang\thanks{This work was supported by the National Natural Science Foundation of China (62103265), the Young Elite Scientists Sponsorship Program by Cast of China Association for Science and Technology (YESS20220320). (\textit{Corresponding author: Zhaojian Wang})}
\thanks{M. Li, Z. Wang, and B. Yang are with the Key Laboratory of System Control, and Information Processing, Ministry of Education of China, Department of Automation, Shanghai Jiao Tong University, Shanghai 200240, China, (email: wangzhaojian@sjtu.edu.cn). }
\thanks{F. Liu is with the State Key Laboratory of Power System and the Department of Electrical Engineering, Tsinghua University, Beijing, 100084, China (e-mail: lfeng@tsinghua.edu.cn).}
\thanks{ M. Cao is with the Institute of Engineering and Technology, University of Groningen, 9747 AG Groningen, Netherlands (e-mail:
 ming.cao@ieee.org).}
}

\maketitle

\begin{abstract}
    The scheme of online optimization as a feedback controller is widely used to steer the states of a physical system to the optimal solution of a predefined optimization problem. Such methods focus on regulating the physical states to the optimal solution in the steady state, without considering the performance during the transient process. In this paper, we simultaneously consider the performance in both the steady state and the transient process of a linear time-invariant system with unknown disturbances. The performance of the transient process is illustrated by the concept of \textit{overtaking optimality}. An overtaking optimal controller with known disturbances is derived to achieve the transient overtaking optimality while guaranteeing steady-state performance. Then, we propose a disturbance independent near-optimal controller, which can achieve optimal steady-state performance and approach the overtaking optimal performance in the transient process. The system performance gap between the overtaking optimal controller and the proposed controller proves to be inversely proportional to the control gains. A case study on a power system with four buses is used to validate the effectiveness of the two controllers.

\end{abstract}


\section{Introduction}

In a physical system, an optimal operation problem is to get the desired steady state after a disturbance, which is usually formulated as an optimization problem. Usually, the problem is solved offline to work out the optimal setpoint, which will be tracked by the physical system under a proper controller with the optimization of transient performance. However, when the timescale of solving the optimization problem is close to that of the physical system dynamics, the interaction between optimization and control has to be considered \cite{ HauswirthAnnuRevControl2024optimization, liu2022merging}. In this situation, performances in the steady state and transient performance need to be optimized simultaneously.

Recently, there have been many researches investigating the interaction of optimization and control in physical systems, which is usually called \textit{feedback-based optimization} \cite{HauswirthAnnuRevControl2024optimization}, or \textit{optimization-guided dynamic control} \cite{wangJAS2023online}. This scheme implements optimization algorithms as dynamical systems and incorporates the algorithms in controllers to realize the merging of optimization and control \cite{zhangTAC2018distributed, menta2018stability, HauswirthTAC2021Timescale, ColombinoTCNS2020Online, bianchinTCNS2022Timevarying, low1999optimization, ChiangIEEEP2007Layering, kasisTSG2017stability, wangTSG2019distributedpart1, liTCNS2016connect}. In \cite{zhangTAC2018distributed}, a distributed control for the network system is proposed to achieve optimal steady-state performance by combining the primal-dual dynamics of the predefined optimization problem with the physical system that can be seen as realizing a primal-dual gradient algorithm. Reference \cite{menta2018stability} relaxes the requirements of the physical system by allowing it to be any linear time-invariant (LTI) system with asymptotic stability. The results in \cite{menta2018stability} are further extended to nonlinear systems in \cite{HauswirthTAC2021Timescale}. It quantifies the required stability bounds for feedback laws based on different optimization algorithms such as gradient descent methods, momentum methods, and primal-dual gradient methods. The time-varying optimization problem in the physical system is considered in \cite{ColombinoTCNS2020Online}, in which a feedback controller based on primal-dual gradient dynamics is designed to steer the physical system to the optimal time-varying trajectory. The tracking error proves to be bounded and the closed-loop system is asymptotically stable. The stability condition checked by solving a linear matrix inequality in \cite{ColombinoTCNS2020Online} is simplified to require the system and control gains to be less than the given constants in \cite{bianchinTCNS2022Timevarying}.  
The scheme has also been widely used in communication networks \cite{low1999optimization, ChiangIEEEP2007Layering} and power systems \cite{liTCNS2016connect, kasisTSG2017stability, wangTSG2019distributedpart1}.

The above works have made significant progress in obtaining the optimal steady state of a dynamical system. \notetomyself{The expression "modeled after" comes from \cite{HauswirthTAC2021Timescale}.}However, the performance during the transient process is not optimized. To improve the transient performance, the reinforcement learning (RL) method is utilized in \cite{cuiTPS2023reinforcement, jiangOJCSYS2022stable, yuanTPS2023learning} by optimizing the controller parameters. Adaptive dynamic programming (ADP) \cite{vrebieAuto2009Adaptive, PangTAC2021Adaptive, PangTAC2023Adaptive} is also an alternative way. Although RL and ADP methods can help improve the transient performances of physical systems, they can hardly provide optimal controllers or optimal performances theoretically.

In fact, obtaining the control law to optimize the transient process relies on the optimization of a functional, which can be obtained by the Pontryagin Maximum Principle \cite{pontryagin2018mathematical} or the Hamilton-Jacobi-Bellman approach \cite{bellman1954theory}. Especially, for an LTI system with quadratic cost, it is the so-called linear quadratic regulator (LQR) problem, where the controller gain can be obtained by solving an algebraic Riccati equation \cite{naidu2002optimal}. However, when tracking a given nonzero signal in the LQR problem, the performance index will be unbounded due to the infinite time horizon. The concept called \textit{overtaking optimality} is a basic tool to describe the optimal controller by comparing its performance index with other controllers, without restricting the performance indices to be bounded (See Section \ref{section_overtaking_optimal_control} for a formal definition of overtaking optimality). By the concept of overtaking optimality, \cite{ArtsteinTAC1985Tracking} extends the infinite-horizon LQR problem into the problem of optimally tracking periodic signals. The result shows that the overtaking optimal controller is a linear feedback of the superposition of the state and a term concerned with the tracking signals. The results in \cite{ArtsteinTAC1985Tracking} are extended in \cite{TANSCL1998Overtaking} and \cite{tanTAC1998nonlinear} by considering the problems to have a linear time-varying system with a quadratic performance index and a nonlinear system with a general performance index, respectively.

By the results in \cite{ArtsteinTAC1985Tracking, TANSCL1998Overtaking, tanTAC1998nonlinear}, the overtaking optimal controller for infinite-horizon time-invariant LQR problems with given tracking signals can be obtained by solving the algebraic Riccati equation. However, if there exist unknown disturbances in the system, the optimal control scheme cannot be applied since the system dynamics is inaccurate. In this paper, the overtaking optimal controller is first derived to achieve optimality in both the steady state and the transient process, requiring information about the disturbances. To eliminate the dependence on unknown disturbances, a near-optimal controller is proposed by incorporating the feedback-based optimization method. System performances under the overtaking optimal controller and the near-optimal controller are then analyzed. The main contributions are summarized as follows.
\begin{itemize}
    \item The closed-form of the overtaking optimal controller for an LTI system with known disturbances is derived. The overtaking optimal controller for the system is shown to be a superposition of a linear state feedback and a constant term concerned with the optimal steady state. The optimal performance of the steady state and transient process can be simultaneously achieved by the overtaking optimal controller.
    \item A near-optimal controller is proposed to steer the LTI system with unknown disturbances to the optimal steady state. It is proved that the physical system is globally asymptotically stable under the proposed near-optimal controller. Although the disturbances are unknown, the equilibrium of the closed-loop system always coincides with the optimal steady state.
    \item We quantify the transient performances of the proposed near-optimal controller and the overtaking optimal controller. The transient performance gap of the proposed controller compared with the optimal value is proved to be zero or inversely proportional to the control gains, depending on the initial state.
\end{itemize}

The rest of the paper is organized as follows. Section II gives the notations, preliminaries, and problem description. The overtaking optimal controller is presented in Section III. Section IV shows the proposed near-optimal controller and its asymptotically stable property. The transient performance indices of the overtaking optimal controller and the proposed near-optimal controller are analyzed in Section V. Simulation results are given in Section VI to validate the effectiveness of the controllers. Section VII concludes the paper.
\section{Problem Formulation}
In this section, we first introduce notations and preliminaries to be used in the rest of this paper. Then, the problem to be solved will be given.
\subsection{Notations}
In this paper, $\mathbb{R}^n$ denotes the $n$-dimensional Euclidean space. The notations $x^T$ and $A^T$ represent the transposes of the vector $x$ and the matrix $A$, respectively. For an invertible matrix $A$, $A^{-T}=(A^{-1})^T$. We use $\mathbf{0}$ to denote the zero matrix or zero vector, and $I$ to denote the identity matrix. When necessary, a subscript will be attached to $\mathbf{0}$ or $I$ to emphasize the dimension. The notation $\text{col}(x_1, x_2, \cdots, x_n)$ represents the stacked vector for $x_1, x_2, \cdots , x_n$, i.e. $\text{col}(x_1, x_2, \cdots, x_n) = \begin{bmatrix}
    x_1^T&x_2^T&\cdots &x_n^T
\end{bmatrix}^T$. The rank of a matrix $A$ is denoted by $\text{rank}(A)$. The nullspace of a matrix $A\in \mathbb{R}^{m\times n}$ is $\text{ker}(A) = \{x\in \mathbb{R}^n| Ax=\mathbf{0}\}$, and the column space of $A$ is $\text{range}(A)=\{y\in \mathbb{R}^m|\exists x\in \mathbb{R}^n \text{ such that } Ax=y\}$. The norm of a vector $x$ is denoted by $\norm{x}=\sqrt{x^Tx}$, and the norm induced by a positive definite matrix $Q$ is represented by $\norm{x}_{Q}=\sqrt{x^TQx}$. When $Q$ is positive semidefinite, $\norm{x}_Q$ is not a norm, but we still use it to represent the result of $\sqrt{x^TQx}$ for convenience.

\subsection{Preliminaries}
\label{subsection_preliminary}
\subsubsection{Controllable and unobservable subspace}
For matrices $A\in \mathbb{R}^{n\times n}$ and $B\in \mathbb{R}^{n\times m}$, the controllable subspace for the system $\dot{x}=Ax+Bu$ is defined as
\begin{align*}
    \mathcal{C}_{AB} := \bigg\{ x\in \mathbb{R}^n\Big|\begin{array}{ll} \exists u(t) \text{ such that }\\
     ~ e^{At} x+ \int_{0}^t e^{A(t-\tau)}Bu(\tau) d\tau =\mathbf{0}\end{array} \bigg\}
\end{align*}
The controllable subspace is equivalent to the column space of the controllability Gramian and the column space of the controllability matrix \cite[Theorem 11.5]{hespanha2018linear}, i.e.,
\begin{align*}
    \mathcal{C}_{AB} &= \text{range}\left(\int_{0}^t e^{-A\tau}BB^Te^{-A^T\tau}d\tau\right)\\
    &=\text{range}\left(\begin{bmatrix}
        B&AB&A^2B& \cdots &A^{n-1}B
    \end{bmatrix}\right)
\end{align*}

For matrices $A\in\mathbb{R}^{n\times n}$ and $C\in \mathbb{R}^{m\times n}$, the unobservable subspace is defined as 
\begin{align}
    \label{eq_unobservable_subspace_definition}
    \mathcal{UO}_{AC} := \left\{ x\in \mathbb{R}^n \big| Ce^{At} x= \mathbf{0} \right\}
\end{align}
It is equivalent to the nullspace of the observability Gramian and the nullspace of the observability matrix \cite[Theorem 15.1]{hespanha2018linear}, i.e.,
\begin{align*}
    \mathcal{UO}_{AC}&=\text{ker} \left(\int_{0}^t e^{A^T\tau}C^TCe^{A\tau} d\tau\right)\\
    &=\text{ker} \left(\begin{bmatrix} C\\ CA\\ \cdots \\ CA^{n-1} \end{bmatrix} \right)
\end{align*}

\subsubsection{Stabilizability and detectability}

For every matrix pair $(A, B)$ with $A \in \mathbb{R}^{n\times n}$ and $B\in \mathbb{R}^{n\times m}$, there exists a nonsingular matrix $T$ which can transform $(A, B)$ into the following form \cite[Theorem 13.2]{hespanha2018linear}.
\begin{align*}
    \tilde{A}&=T^{-1}AT=\begin{bmatrix}
        A_c&A_{12}\\ \mathbf{0}&A_{u}
    \end{bmatrix}\\
    \tilde{B}&=T^{-1}B=\begin{bmatrix}
        B_c\\ \mathbf{0}
    \end{bmatrix}
\end{align*}
where $(A_c, B_c)$ is controllable. Then $(A, B)$ is said to be stabilizable if the matrix $A_u$ is Hurwitz, i.e. all of the eigenvalues of $A_u$ have negative real parts. If $(A,B)$ is stabilizable, then the matrix $\begin{bmatrix}A&B\end{bmatrix}$ is of full row rank, i.e. $\text{rank}(\begin{bmatrix}A&B\end{bmatrix}) = n$. For matrices $A\in \mathbb{R}^{n\times n}$ and $C\in \mathbb{R}^{m\times n}$, $(A,C)$ is called detectable if $(A^T, C^T)$ is stabilizable.


\subsubsection{Lyapunov equation and algebraic Riccati equation}
\label{subsection_lyapunov_and_ARE_equation}
For a Hurwitz matrix $A$ and a positive semidefinite matrix $Q$, the Lyapunov equation $PA+A^TP+Q=\mathbf{0}$ has a unique solution $P^*$. The solution $P^*$ can be represented by $P^*=\int_0^{\infty}e^{A^Tt}Qe^{At}dt$ and it is positive semidefinite. Besides, $P^*$ is positive definite if and only if $(A, Q)$ is observable\cite[Theorem 12.4]{hespanha2018linear}. 

If $(A, B)$ is stabilizable and $(A, Q)$ is detectable for the positive semidefinite matrix $Q$, then for any positive definite matrix $R$, the algebraic Riccati equation $PA+A^TP-PBR^{-1}B^TP+Q=\mathbf{0}$ has a unique positive semidefinite solution $P^*$ and $A-BR^{-1}B^TP^*$ is Hurwitz \cite[Theorem 5]{kuvcera1973review}. Further, if $(A, Q)$ is observable, $P^*$ is positive definite.

\subsubsection{$\mathcal{L}$-stabliity} 
The space $\mathcal{L}_p$ for $1\le p <\infty$ is defined as 
\begin{align*}
    \mathcal{L}_p := \left\{u\ \big|\left(\int_{0}^{\infty} \norm{u(t)}^pdt  \right)^{1/p}\le \infty\right\}
\end{align*}
Specially, $\mathcal{L}_{\infty}$ is defined as $\mathcal{L}_{\infty} = \{u(t)|\sup_{t\ge 0} \norm{u(t)}<\infty\}$, which represents the space of piecewise continuous, bounded functions \cite[Chapter 5.1]{2002Nonlinear}. For the system $\dot{x} = Ax+Bu$, if $A$ is Hurwitz and $u\in \mathcal{L}_2$, then $x(t)\to \mathbf{0}$ as $t\to \infty$ \cite[Appendix A.5]{mareels1996adaptive}.

\subsection{Problem description}\label{section_overtaking_optimal_control}
We consider a system with the following dynamics
\begin{align}
    \label{eq_physical_system}
    \dot{x} &= Ax+Bu+Cd \todofuture{+\epsilon}
\end{align}
where $x\in \mathbb{R}^n$ is the system state, $u\in \mathbb{R}^m$ is the input, $d\in \mathbb{R}^p$ is the unknown disturbance, and matrices $A, B$ and $C$ have proper dimensions.

The optimization problem for the system in the steady state is as follows.
\begin{subequations}
    \label{eq_steady_optimization_problem}
    \begin{align}
        \label{eq_steady_optimization_problem_target}
        \min_{y, u} &\quad  f(y, u) := \frac{1}{2} (y^TQy+u^TRu)\\
        \label{eq_steady_optimization_problem_constraint}
        s.t.&\quad Ay+Bu+Cd=\mathbf{0}
    \end{align}
\end{subequations}
where the variable $y\in \mathbb{R}^{n}$ represents the steady state of $x$, which is used here to avoid notation confusion. $Q\in \mathbb{R}^{n\times n}$ is positive semidefinite, $R\in \mathbb{R}^{m\times m}$ is positive definite, and \eqref{eq_steady_optimization_problem_constraint} represents system \eqref{eq_physical_system} is in the steady state. 

The transient performance index to be minimized is
\begin{align}
    \label{eq_performance_index_J}
    J_{\infty}(x_0, u(t)) := \int_{0}^\infty f\left(x(t),u(t)\right) dt
\end{align}
where $x_0$ is the initial state, and $f(x,u)$ is the cost function defined in \eqref{eq_steady_optimization_problem_target}.

We make the following assumption about the system and the performance index, which is very common.
\begin{assumption}
    \label{assumption_system_property}
    The matrix pair $(A, B)$ is stabilizable, and $(A, Q)$ is detectable.
\end{assumption}

We will focus on \textit{admissible} controllers satisfying the following definition.

\begin{definition}
    For a controller $u(t)$ and the corresponding trajectory $x(t)$, $u(t)$ is said to be admissible if $u(t)\in \mathcal{L}_{\infty}$ and $x(t)\in \mathcal{L}_{\infty}$.
\end{definition}
The set for all admissible controllers is denoted by $\mathcal{U}$.

The transient performance index $J_{\infty}(x_0,u(t))$ can be unbounded for every admissible controller due to the existence of disturbance $d$ in the system. In this situation, we adopt the concept of \textit{overtaking optimality} in Definition \ref{Def_overtaking_optimal} to differentiate the performance of different controllers.

\begin{definition}\cite[Definition II.2]{tanTAC1998nonlinear}\label{Def_overtaking_optimal}
    Define $J_{T}(x_0, u) := \int_{0}^{T} f(x(t),u(t)) dt$. The controller $u^*(t)$ is said to be overtaking optimal if the following inequality in \eqref{eq_definition_overtaking_optimal} holds for any admissible control $u(t)\in \mathcal{U}$.
    \begin{equation}
        \label{eq_definition_overtaking_optimal}
        \limsup_{T \to \infty} \{J_T(x_0, u^*)-J_{T}(x_0, u)\}\le 0
    \end{equation}
\end{definition}

In the rest of the paper, we will design admissible controllers to achieve the following two goals simultaneously. 
\begin{itemize}
    \item \textbf{P1}: Steer system \eqref{eq_physical_system} to the optimal solution to problem \eqref{eq_steady_optimization_problem};
    \item \textbf{P2}: Minimize the performance index in \eqref{eq_performance_index_J} in the sense of overtaking optimality.
\end{itemize}



\section{Overtaking optimal controller}
Before giving the overtaking optimal controller, we first investigate the steady-state optimization problem in \eqref{eq_steady_optimization_problem}. Consider the following Lagrangian function
\begin{align}
    \label{eq_lagrangian}
    L(y, u, \lambda) = \frac{1}{2}(y^TQy+u^TRu)+\lambda^T(Ay+Bu+Cd)
\end{align}
where $\lambda\in \mathbb{R}^n$ is the Lagrangian multiplier for the constraint in \eqref{eq_steady_optimization_problem_constraint}. The properties of the optimal solution and the Lagrangian multiplier are given in the following lemma.

\begin{lemma}
    \label{lemma_steady_state_problem}
    Suppose Assumption \ref{assumption_system_property} holds. For any fixed disturbance $d$, the following results hold.
    \begin{enumerate}
        \item There exists a unique optimal solution $(\bar{x}, \bar{u})$ to the problem in \eqref{eq_steady_optimization_problem}.
        \item $(\bar{x}, \bar{u})$ is the optimal solution if and only if there exists a unique solution to $\lambda = \bar{\lambda}\in \mathbb{R}^{n}$ satisfying
            \begin{subequations}
                \label{eq_KKT_condition_static_OP}
                \begin{align}
                    \label{eq_KKT_condition_static_OP_1}
                    Q\bar{x}+A^T {\lambda}&=\mathbf{0}\\
                    \label{eq_KKT_condition_static_OP_2}
                    R\bar{u}+B^T{\lambda}&=\mathbf{0}\\
                    \label{eq_KKT_condition_static_OP_3}
                    A\bar{x}+B\bar{u}+Cd&=\mathbf{0}
                \end{align}
            \end{subequations}
    \end{enumerate}
\end{lemma}

\begin{proof} For the first assertion, the existence is proved as follows. Since $(A, B)$ is stabilizable, $\text{rank}\left(\begin{bmatrix}A&B\end{bmatrix}\right)=n$. Therefore, the feasible set $\{(y, u)|Ay+Bu+Cd=\mathbf{0}\}$ is not empty since $\text{rank}\left(\begin{bmatrix}A&B\end{bmatrix}\right)=\text{rank}\left(\begin{bmatrix}A&B&Cd\end{bmatrix}\right)$ holds.  We then show the optimal solution is bounded. Since $R$ is positive definite, $\bar{u}$ is bounded. If $\bar{x}$ is unbounded, then it can be expressed as $\bar{x}=k\alpha$ where $\alpha\in \mathbb{R}^{n}, \norm{\alpha} = 1$, $k\in\mathbb{R}$ and $k\to \infty$. Since $A\bar{x}+B\bar{u}+Cd=\mathbf{0}$, it must hold that $A\alpha=\mathbf{0}$. On the other hand, the fact that $\bar{x}^TQ\bar{x}$ is bounded implies $Q{\alpha}=\mathbf{0}$. Thus, $\alpha$ is in the nullspace of  $\begin{bmatrix}
        Q\\A
    \end{bmatrix}$. Since $(A, Q)$ is detectable, $\begin{bmatrix}
        Q\\A
    \end{bmatrix}$ is of full column rank, which means $\alpha$ must be $\mathbf{0}$ and this contradicts with the fact of $\norm{\alpha}=1$. Therefore, there exists an optimal solution.
    
    To prove the uniqueness of the solution, suppose $(\hat{x}, \hat{u})$ is also the optimal solution. Let $(y, u)=(\frac{1}{2}(\bar{x}+\hat{x}), \frac{1}{2}(\bar{u}+\hat{u}))$ and we have
    \begin{subequations}
        \begin{align*}
            &  f(\bar{x}, \bar{u})+f(\hat{x}, \hat{u})-2f(y,u)\\
            &\quad=\frac{1}{4}(\bar{x}-\hat{x})^T Q (\bar{x}-\hat{x}) + \frac{1}{4}(\bar{u}-\hat{u})^TR(\bar{u}-\hat{u})
        \end{align*}
    \end{subequations}
    Since $(\bar{x},\bar{u})$ and $(\hat{x}, \hat{u})$ are both the optimal solutions, we have
    \begin{align*}
        \frac{1}{4}(\bar{x}-\hat{x})^T Q (\bar{x}-\hat{x}) + \frac{1}{4}(\bar{u}-\hat{u})^TR(\bar{u}-\hat{u})\le 0
    \end{align*}
    Therefore, $(\bar{x}-\hat{x})^T Q (\bar{x}-\hat{x})=0$ and $\bar{u}=\hat{u}$ hold. By $(\bar{x}-\hat{x})^T Q (\bar{x}-\hat{x})=0$, we have $Q(\bar{x}-\hat{x})=\mathbf{0}$, i.e. $\hat{x}-\bar{x}\in \ker(Q)$. Therefore, there exists $\alpha\in \ker(Q)$ such that $\hat{x}=\bar{x}+\alpha$. Since $(\bar{x}, \bar{u})$ and $(\hat{x}, \hat{u})$ satisfy \eqref{eq_steady_optimization_problem_constraint}, we have
    \begin{align*}
        A\alpha &= A(\hat{x}-\bar{x})\\
        &=(A\hat{x}+B\hat{u}+Cd)-(A\bar{x}+B\bar{u}+Cd)\\
        &=\mathbf{0}
    \end{align*}
    which means that $\alpha$ satisfies $\begin{bmatrix}
        Q\\A
    \end{bmatrix}\alpha=\mathbf{0}$. By the full column rank property of $\begin{bmatrix}
        Q\\A
    \end{bmatrix}$, $\alpha=\mathbf{0}$ is the unique solution, i.e. $\bar{x}=\hat{x}$. We conclude the optimal solution $(\bar{x}, \bar{u})$ exists uniquely.

    For the second assertion, the KKT conditions for problem \eqref{eq_steady_optimization_problem} are 
    \begin{align}
        \frac{\partial L}{\partial y} &= \mathbf{0}, &   \frac{\partial L}{\partial u} &=\mathbf{0}, &\frac{\partial L}{\partial \lambda} &=\mathbf{0}
    \end{align}
    which are exactly the conditions in \eqref{eq_KKT_condition_static_OP}. Since problem \eqref{eq_steady_optimization_problem} is a convex optimization problem with linear equality constraints, the KKT conditions are sufficient and necessary for the optimal solution \cite[Chapter 5.5.3]{boyd2004convex}. The uniqueness of $\bar{\lambda}$ can be proved as follows. By \eqref{eq_KKT_condition_static_OP_1} and \eqref{eq_KKT_condition_static_OP_2}, we have 
    \begin{align}
        \label{eq_unique_lambda}
        \begin{bmatrix}
            A^T\\ B^T
        \end{bmatrix}{\lambda}=\begin{bmatrix}
            -Q\bar{x}\\-R\bar{u}
        \end{bmatrix}
    \end{align} Since $\begin{bmatrix}
        A^T\\B^T
    \end{bmatrix}$ is of full column rank, equation \eqref{eq_unique_lambda} has at most one solution. Therefore, $\bar{\lambda}$ is the unique solution.
\end{proof}

For the optimization of the transient performance index $J_{\infty}$, the following theorem is given to obtain the overtaking optimal controller $u^*(t)$.
\begin{theorem}
    \label{theorem_overtaking_optimal}
    Suppose Assumption \ref{assumption_system_property} holds. Let $(\bar{x}, \bar{u})$ be the optimal solution to the optimization problem \eqref{eq_steady_optimization_problem}. The overtaking optimal controller is 
    \begin{equation}
        \label{eq_optimal_control_trace}
        u^*(t)=-K(x(t)-\bar{x})+\bar{u}
    \end{equation}
    where $K=R^{-1}B^TP^*$ with $P^*$ being the unique positive semidefinite solution to the following algebraic Riccati equation
    \begin{equation}
        \label{eq_Riccati}
        PA+A^TP-PBR^{-1}B^TP+Q=\mathbf{0}
    \end{equation}
\end{theorem}

The proof of Theorem \ref{theorem_overtaking_optimal} is given in Appendix. 
Theorem \ref{theorem_overtaking_optimal} implies that the overtaking optimal control is the superposition of $\bar{u}$ and the linear feedback of $x-\bar{x}$. If no disturbance exists in the system, we have $\bar{x}=\mathbf{0}$ and $\bar{u}=\mathbf{0}$. Then, the overtaking optimal controller degrades to a common LQR controller. 

We have the following convergence result of the closed-loop system under the overtaking optimal controller.

\begin{lemma}
    \label{lemma_property_of_overtaking_optimal_controller}
    The closed-loop system composed of \eqref{eq_physical_system} and \eqref{eq_optimal_control_trace} is asymptotically stable. The equilibrium for the closed-loop system is $\bar{x}$, and $u^*(t)$ converges to $\bar{u}$, i.e., the optimal solution to the optimization problem \eqref{eq_steady_optimization_problem}.
\end{lemma}
\begin{proof}
    Applying the overtaking optimal controller \eqref{eq_optimal_control_trace} into system \eqref{eq_physical_system}, the closed-loop system becomes
    \begin{align}
        \dot{x}&=(A-BK)x+BK\bar{x}+B\bar{u} +Cd\nonumber\\
        \label{eq_closed_loop_system_A_minus_BK}
        &=(A-BK)x-(A-BK){\bar{x}}
    \end{align}
    where the fact $\bar{u}+Cd=-A\bar{x}$ in \eqref{eq_KKT_condition_static_OP_3} is used. By the preliminaries of the algebraic Riccati equation in Section \ref{subsection_preliminary}, $A-BK$ is Hurwitz. Therefore, system \eqref{eq_closed_loop_system_A_minus_BK} is asymptotically stable. Since $x=\bar{x}$ is the equilibrium of system \eqref{eq_closed_loop_system_A_minus_BK}, $x(t)$ converges to $\bar{x}$. By $u^*(t)=-K(x(t)-\bar{x})+\bar{u}$, $u^*(t)$ converges to $\bar{u}$. This completes the proof.
\end{proof}

Lemma \ref{lemma_property_of_overtaking_optimal_controller} shows that the overtaking optimal controller $u^*(t)$ can steer the system to the optimal steady state, which, together with Theorem \ref{theorem_overtaking_optimal}, implies that $u^*(t)$ can realize \textbf{P1} and \textbf{P2} simultaneously. However, since $d$ is an unknown disturbance, the optimal steady state $\bar{x}$ and input $\bar{u}$ cannot be obtained in \eqref{eq_optimal_control_trace}, for which we propose a near-optimal controller independent of the disturbance in the following section.

\section{Near-optimal Controller}
In this section, we design a near-optimal controller to remove the dependence on the disturbance. The controller is shown to be capable of steering the system to the optimal steady state asymptotically.

\subsection{Near-optimal controller design}

Recalling the Lagrangian function $L(y, u, \lambda)$ in \eqref{eq_lagrangian}, the pair of the primal optimal solution $(\bar{x}, \bar{u})$ and the dual optimal solution $\bar{\lambda}$ is the saddle point of $L(y, u, \lambda)$ according to Lagrange duality \cite[Chapter 5.4]{boyd2004convex}. The primal optimal solution to $L(y, u, \lambda)$ about $u$ can be obtained by
\begin{align}
    \label{eq_h_definition}
    h(y, \lambda):=\arg\mathop{\min}_{u} ~L(y, u, \lambda)=-R^{-1} B^T \lambda
\end{align}
Substituting $u=h(y, \lambda)$ into the Lagrangian function $L(y, u, \lambda)$, we can obtain a reduced Lagrangian function $\tilde{L}\left(y, \lambda\right)$ as follows.
\begin{align}
    \label{eq_reduced_Lagrangian}
    \begin{split}
        \tilde{L}(y, \lambda)&=L(y, h(y, \lambda), \lambda)\\
        &=\frac{1}{2}\left(y^TQy+\lambda^TBR^{-1}B^T\lambda\right) \\
        & \quad \quad +\lambda^T(Ay-BR^{-1}B^T \lambda+Cd)
    \end{split}
\end{align}
The saddle point of $\tilde{L}(y, \lambda)$ is precisely the optimal primal-dual pair $(\bar{x}, \bar{u}, \bar{\lambda})$ by combining $u=h(y, \lambda)$. The primal-dual dynamics to seek the saddle point of $\tilde{L}(y, \lambda)$ is 
\begin{subequations}
    \label{eq_primal_dual_dynamics}
    \begin{alignat}{4}
        \label{eq_primal_dual_dynamics_y}
        \dot{y} &=-K^y\frac{\partial \tilde{L}}{\partial y}&&= -K^y (Qy+A^T\lambda)\\
        \label{eq_primal_dual_dynamics_lambda}
        \dot{\lambda} &=K^{\lambda}\frac{\partial \tilde{L}}{\partial \lambda} &&= K^{\lambda} (Ay-BR^{-1}B^T \lambda +Cd)
    \end{alignat}
\end{subequations}
where $K^y$ and $K^{\lambda}$ are positive definite matrices. 

The primal-dual dynamics in \eqref{eq_primal_dual_dynamics} and the overtaking optimal controller in \eqref{eq_optimal_control_trace} motivate the following near-optimal feedback controller.
\begin{equation}
    \label{eq_controller_lambda}
    u = -K(x-y)-R^{-1}B^T \lambda
\end{equation}
where $\bar{x}$ in \eqref{eq_optimal_control_trace} is replaced by $y$, and $\bar{u}$ is replaced by $-R^{-1}B^T\lambda$. 

In practical implementation, \eqref{eq_primal_dual_dynamics_lambda} cannot be realized since the information $d$ is needed. Performing a linear transformation $\mu = \lambda-K^{\lambda} x$ to the controller \eqref{eq_primal_dual_dynamics} and \eqref{eq_controller_lambda}, we can obtain the following closed-loop system 
\begin{subequations}
    \label{eq_closed_loop_system_mu}
    \begin{align}
        \dot{x} & = Ax+Bu+Cd\\
        \label{eq_controller_mu_1}
        \dot{y} & = -K^y\left(Qy+A^T(\mu+K^\mu x)\right)\\
        \label{eq_controller_mu_2}
        \dot{\mu} & = K^{\mu} (A-BK)(x-y)\\
        \label{eq_controller_mu_3}
        u&=-K(x-y)-R^{-1}B^T(\mu+K^\mu x)
    \end{align}
\end{subequations}
where $K^{\mu}$ is copy of $K^{\lambda}$, i.e. $K^{\mu}=K^{\lambda}$. Obviously, the controller in the closed-loop system \eqref{eq_closed_loop_system_mu} does not rely on the information of the disturbance $d$. Thus, the controller in system \eqref{eq_closed_loop_system_mu} is implementable.

\subsection{Asymptotic stabliity}
In this subsection, we prove the asymptotic stability of the closed-loop system \eqref{eq_closed_loop_system_mu}. 

We first write the closed-loop system under the controller in \eqref{eq_primal_dual_dynamics} and \eqref{eq_controller_lambda} as the following form
\begin{subequations}
    \label{eq_closed_loop_system_lambda}
    \begin{align}
        \dot{x} & = Ax-BK(x-y)+BR^{-1}B^T\lambda+Cd\\
        \dot{y} & = -K^y (Qy+A^T\lambda)\\
        \dot{\lambda} & = K^{\lambda} (Ay-BR^{-1}B^T\lambda+Cd)
    \end{align}
\end{subequations}
Because system \eqref{eq_closed_loop_system_mu} can be obtained by the nonsingular linear transformation $\mu=\lambda-K^{\lambda} x$ from \eqref{eq_closed_loop_system_lambda}, the following lemma is straightforward.
\begin{lemma}
    \label{lemma_equivalence_mu_lambda}
    The closed-loop system \eqref{eq_closed_loop_system_mu} is equivalent to \eqref{eq_closed_loop_system_lambda} if the initial state $(x^{\mu}(0), y^{\mu}(0), \mu^{\mu}(0))$ for system \eqref{eq_closed_loop_system_mu} and the initial state $(x^{\lambda}(0), y^{\lambda}(0), \lambda^{\lambda}(0))$ satisfy the conditions in \eqref{eq_initial_state_condition} below
        \begin{align}
            \label{eq_initial_state_condition}
            \begin{split}
                x^{\mu}(0)&=x^{\lambda}(0)\\
                y^{\mu}(0)&=y^{\lambda}(0)\\
                \mu^{\mu}(0)&=\lambda^{\lambda}(0)-K^{\lambda} x^{\lambda}(0)
            \end{split}
        \end{align}
    In this situation, the following results hold for any $t\ge0$
    \begin{align}
        \begin{split}
            x^{\mu}(t)&=x^{\lambda}(t)\\
            y^{\mu}(t)&=y^{\lambda}(t) \\
            \mu^{\mu}(t)&=\lambda^{\lambda}(t)-K^{\lambda} x^{\lambda}(t)
            \end{split}
        \end{align}
\end{lemma}

Therefore, if we can prove the asymptotic stability of system \eqref{eq_closed_loop_system_lambda}, then the asymptotic stability of system \eqref{eq_closed_loop_system_mu} can be immediately obtained. For the same reason, we will focus on analyzing the behaviors of system \eqref{eq_closed_loop_system_lambda} instead of \eqref{eq_closed_loop_system_mu} for convenience in the rest of the paper.

To analyze the asymptotic stability of dynamics \eqref{eq_primal_dual_dynamics}, we introduce the compact form of \eqref{eq_primal_dual_dynamics} as follows.
    \begin{align}
        \label{eq_primal_dual_dynamics_compact}
        \begin{bmatrix}
            \dot{y}\\ \dot{\lambda}
        \end{bmatrix}&=S\begin{bmatrix}
            y\\ \lambda
        \end{bmatrix}+\begin{bmatrix}
            \mathbf{0}\\ K^{\lambda} Cd
        \end{bmatrix}
    \end{align}
where $S$ is the state matrix defined as
\begin{align}
    \label{eq_S_definition}
    S := K^S T
\end{align}
with 
\begin{subequations}
        \begin{align}
            \label{eq_K_S_def}
            K^S&:=\begin{bmatrix}
            K^y&\mathbf{0}\\\mathbf{0}&K^{\lambda}
            \end{bmatrix}\\
            \label{eq_T_def}
            T&:=\begin{bmatrix}
                -Q&-A^T\\ A&-BR^{-1}B^T
            \end{bmatrix}
        \end{align}
    \end{subequations}
The following lemma is given to show the stability of the primal-dual dynamics in \eqref{eq_primal_dual_dynamics_compact}.
\begin{lemma}
    \label{lemma_Hurwitz_property_T_S}
    Suppose $Q$ is positive definite. Then the following results hold:
    \begin{enumerate}
    \item The state matrix $S$ is Hurwitz;
    \item The primal-dual dynamics in \eqref{eq_primal_dual_dynamics} converges asymptotically with $y(t)\to \bar{x}, \lambda(t)\to \bar{\lambda}$ and $ -R^{-1}B^T\lambda(t)\to \bar{u}$ as $t\to \infty$.
    \end{enumerate}
\end{lemma}
\begin{proof}
    To prove the first assertion, consider a system $\dot{x}=Sx$ with $x=\text{col}(x_1, x_2), x_1, x_2\in\mathbb{R}^n$. Take the Lyapunov function $V(x) = \frac{1}{2} x^T (K^{S})^{-1} x$, and we have
    \begin{align}
        \dot{V}=-x_1^TQx_1-x_2^TBR^{-1}B&^Tx_2\le 0
    \end{align}
    By $\dot{V}(x)=0$, we have $x_1 \equiv \mathbf{0},B^Tx_2 \equiv\mathbf{0}$. Since $ \dot{x}_1 =-Qx_1-A^Tx_2\equiv\mathbf{0}$, we can know that $A^Tx_2\equiv \mathbf{0}$. Therefore, $\begin{bmatrix}
        B^T\\A^T
    \end{bmatrix}x_2\equiv \mathbf{0}$. Since $(A, B)$ is stabilizable, $\begin{bmatrix}
        B^T\\A^T
    \end{bmatrix}$ is of full column rank, from which we have $x_2\equiv \mathbf{0}$. Therefore, the system is asymptotically stable by LaSalle's invariance principle \cite[Theorem 4.4]{2002Nonlinear}, which implies the matrix $S$ is Hurwitz.
    

    Substituting $(\bar{x}, \bar{\lambda})$ into the righthand side of \eqref{eq_primal_dual_dynamics_compact}, we have
    \begin{align*}
        S\begin{bmatrix}
            \bar{x}\\ \bar{\lambda}
        \end{bmatrix}+\begin{bmatrix}
            \mathbf{0}\\ K^{\lambda} Cd
        \end{bmatrix}&=\begin{bmatrix}
            -K^y(Q\bar{x}+A^T\bar{\lambda})\\
            K^{\lambda}(A\bar{x}-BR^{-1}B^T\bar{\lambda}+Cd)
        \end{bmatrix}\\
        &=\begin{bmatrix}
            \mathbf{0}\\ \mathbf{0}
        \end{bmatrix}
    \end{align*}
    where the second equality holds due to the conditions in \eqref{eq_KKT_condition_static_OP}. Therefore, we can know that $(\bar{x}, \bar{\lambda})$ is the equilibrium of the primal-dual system. By the Hurwitz property of $S$, $y(t)$ and $\lambda(t)$ converge to $\bar{x}$ and $\bar{\lambda}$, respectively. Then we can know $-R^{-1}B^T\lambda(t)$ converges to $\bar{u}$ by the fact $\bar{u}=-R^{-1}B^T\bar{\lambda}$ derived in \eqref{eq_KKT_condition_static_OP_2}. This completes the proof.
\end{proof}

\begin{remark}
    In Lemma \ref{lemma_Hurwitz_property_T_S}, the positive definiteness of $Q$ is a sufficient condition for the asymptotic stability of the primal-dual dynamics. Even if $Q$ is \textit{positive semidefinite}, the analysis to follow in the rest of the paper still holds as long as the primal-dual dynamics is asymptotically stable under other conditions. The case in Section \ref{section_case_study} will provide such an example, where $Q$ is semidefinite while the primal-dual dynamics is asymptotically stable. Generally speaking, the asymptotic stability can be satisfied by many practical examples, which has been investigated by many existing works, e.g. \cite{CortsJNS2016Distributed, ChangCSLetter2019Saddleflow, bianchinTCNS2022Timevarying, TangSIAM2022Running}. Since the primal-dual dynamics \eqref{eq_primal_dual_dynamics_compact} is a linear system, the asymptotic stability is equivalent to the Hurwitz property of $S$. Therefore, we make the following assumption to describe the asymptotic stability of the primal-dual dynamics.
\end{remark}

\begin{assumption}
    \label{assumption_hurwitz_S}
    The matrix $S$ is Hurwitz.
\end{assumption}

Then we are ready for the asymptotic stability of system \eqref{eq_closed_loop_system_lambda}, which indicates the stability of the implementable closed-loop system \eqref{eq_closed_loop_system_mu} by the equivalence property in Lemma \eqref{lemma_equivalence_mu_lambda}.

\begin{theorem}
    \label{theorem_asymptotic_stablity}
    Suppose Assumption \ref{assumption_system_property} and \ref{assumption_hurwitz_S} hold. The closed-loop system \eqref{eq_closed_loop_system_lambda} is asymptotically stable with the unique equilibrium of the system as $(x,y,\lambda)=(\bar{x},\bar{x},\bar{\lambda})$.
\end{theorem}
\begin{proof}
    The state matrix of \eqref{eq_closed_loop_system_lambda} is 
    \begin{subequations}
        \label{eq_tilde_A}
        \begin{align}
                \tilde{A}&=\begin{bmatrix}
                        A-BK & BK&-BR^{-1}B^T\\
                        \mathbf{0}&-K^{y}Q&-K^y A^T\\
                        \mathbf{0}&K^{\lambda}A&K^{\lambda}BR^{-1}B^T
                    \end{bmatrix}\\
                &=\begin{bmatrix}
                    A-BK & N\\ \mathbf{0}&S
                \end{bmatrix}
        \end{align}
    \end{subequations}
    where $N$ is defined as $N:=\begin{bmatrix}
        BK&-BR^{-1}B^T
    \end{bmatrix}$. The eigenvalues of the matrix $\tilde{A}$ are the collection of the eigenvalues of $A-BK$ and $S$ by the upper triangular form of $\tilde{A}$. Since $A-BK$ and $S$ are both Hurwitz, all the eigenvalues of $\tilde{A}$ have negative real parts, which means $\tilde{A}$ is Hurwitz and the system is asymptotically stable. Then it can be easily verified that $(\bar{x},\bar{x},\bar{\lambda})$ is the unique equilibrium of system \eqref{eq_closed_loop_system_lambda}. 
\end{proof}

Theorem \ref{theorem_asymptotic_stablity} shows that the near-optimal controller can steer the system to the optimal steady state although the disturbance is unknown, i.e., realizing \textbf{P1}. For \textbf{P2}, we will analyze the performance gap of the system under the overtaking optimal controller and the proposed optimal controller in the following section.

\section{Transient Performance Analysis}
In this section, we first introduce the transient performance index for an asymptotically stable LTI system in Section \ref{subsection_integral_index}. Then, the system performance indices under the overtaking optimal controller and the proposed near-optimal controller are derived in Section \ref{subsection_optimal_performance} and \ref{subsection_near_optimal_performance}, respectively. Section \ref{subsection_discussion} presents discussions about the transient system performances under the two controllers.

\subsection{Performance index}
\label{subsection_integral_index}
The transient performance index $J_{\infty}$ in \eqref{eq_performance_index_J} is an integral over the quadratic function $f(x,u)=\frac{1}{2}(x^TQx+u^TRu)$. Because $u$ is the affine function of $x$ in controllers \eqref{eq_optimal_control_trace} and \eqref{eq_controller_mu_3}, $f(x, u)$ will be quadratic polynomials about the state variables. Then the performance indices will be the integrals over the quadratic polynomials as we will see in Section \ref{subsection_optimal_performance} and \ref{subsection_near_optimal_performance}. To this end, we first investigate the general linear system in \eqref{eq_general_linear_system} with two types of performance index $\mathcal{J}_{\infty}^{\mathcal{Q}}, \mathcal{J}_{\infty}^{\alpha}$ in \eqref{eq_performance_type_Q}, \eqref{eq_performance_type_alpha}, which will be utilized in the computation of the system performance indices under the two controllers.
\begin{subequations}
\begin{align}
    \label{eq_general_linear_system}
    \dot{\psi} &= \mathcal{A}\psi + \beta\\
    \label{eq_performance_type_Q}
    \mathcal{J}_{\infty}^{\mathcal{Q}} &:= \int_{0}^{\infty} \psi^T\mathcal{Q}\psi dt\\
    \label{eq_performance_type_alpha}
    \mathcal{J}_{\infty}^{\alpha}&:=\int_{0}^{\infty} \alpha^T \psi dt
\end{align}    
\end{subequations}
where $\psi\in \mathbb{R}^{l}$ is the state variable for some $l>0$, $\beta \in \mathbb{R}^{l}, \alpha\in \mathbb{R}^{l}$ are constant vectors, $\mathcal{A}\in \mathbb{R}^{l\times l}$ is Hurwitz, and $\mathcal{Q}\in \mathbb{R}^{l\times l}$ is positive semidefinite. 

Since the $\mathcal{A}$ is Hurwitz, the state $\psi(t)$ in system \eqref{eq_general_linear_system} will converge to $\mathcal{A}^{-1}\beta$ as time goes to infinity. Therefore, $\mathcal{J}_{\infty}^{\mathcal{Q}}$ and $\mathcal{J}_{\infty}^{\alpha}$ can be unbouned. For convenience, we define the unbounded integral over the infinite time horizon as follows.
\begin{definition}
    \label{definition_unbounded_integral}
    If $\int_0^{\infty}f(t)dt$ is unbounded while $\int_{0}^{\infty} (f(t)-a)dt$ is bounded, the unbounded integral $\int_{0}^{\infty} f(t)dt$ is represented as $\int_{0}^{\infty} (f(t)-a)dt+a\intoninf$, where the notation $\intoninf$ can be treated as the unbounded integral $\int_0^{\infty} 1dt$. 
\end{definition}

Then, we use the notation of $\intoninf$ in Definition \ref{definition_unbounded_integral} to study the performance indices $\mathcal{J}_{\infty}^{\mathcal{Q}}$ and $\mathcal{J}_{\infty}^{\alpha}$. We have the following result.
\begin{lemma}
    \label{lemma_universal_performance_computation}
   For system \eqref{eq_general_linear_system} and $\mathcal{J}_{\infty}^{\mathcal{Q}}$ and $\mathcal{J}_{\infty}^{\alpha}$ defined in \eqref{eq_performance_type_Q} and \eqref{eq_performance_type_alpha}, we have the following results.
    \begin{enumerate}
        \item For the quadratic cost $\mathcal{J}_{\infty}^{\mathcal{Q}}$, it is computed by
        \begin{align}
            \label{eq_J_Q}
            \begin{split}
                \mathcal{J}_{\infty}^{\mathcal{Q}}&= \norm{\psi_0+\mathcal{A}^{-1}\beta}_{\mathcal{P}}^2\\
                &\quad +2\beta^T\mathcal{A}^{-T}\mathcal{Q}\mathcal{A}^{-1}(\psi_0+\mathcal{A}^{-1}\beta)\\
                &\quad +\norm{\mathcal{A}^{-1}\beta}_{\mathcal{Q}}^2 \intoninf
            \end{split}
        \end{align}
        where $\psi_0$ is the initial state at $t=0$, $\mathcal{P}$ is the solution to the Lyapunov equation $\mathcal{P}\mathcal{A}+\mathcal{A}^T\mathcal{P}+\mathcal{Q}=\mathbf{0}$.
        \item For the linear cost $\mathcal{J}_{\infty}^{\alpha}$, it is computed by
        \begin{align}
            \label{eq_J_alpha}
            \begin{split}
                \mathcal{J}_{\infty}^{\alpha}&=-\alpha^T \mathcal{A}^{-1}(\psi_0+\mathcal{A}^{-1}\beta)-\alpha^T \mathcal{A}^{-1}\beta \intoninf
            \end{split}
        \end{align}
    \end{enumerate}
\end{lemma}
\begin{proof}
    The solution to $\dot{\psi}=\mathcal{A}\psi+\beta$ is 
    \begin{equation}
        \label{eq_psi_t}
        \psi(t)=e^{\mathcal{A}t}(\psi_0+\mathcal{A}^{-1}\beta)-\mathcal{A}^{-1}\beta
    \end{equation}

    Substituting $\psi(t)$ into $\int_{0}^{\infty} \psi^T\mathcal{Q}\psi dt$, we have
    \begin{subequations}\label{General_JQ}
        \begin{align}
            &\quad \int_{0}^{\infty} \psi^T\mathcal{Q}\psi dt\nonumber\\
            &= \int_{0}^{\infty} \big((\psi_0^T +\beta^T\mathcal{A}^{-T})e^{\mathcal{A}^Tt} -\beta^T\mathcal{A}^{-T}\big)\mathcal{Q}\big(e^{\mathcal{A}t}(\psi_0 \nonumber\\
            &\qquad \qquad \qquad \qquad+ \mathcal{A}^{-1}\beta)-\mathcal{A}^{-1}\beta\big) dt\nonumber\\
            &=\int_{0}^{\infty} (\psi_0^T +\beta^T\mathcal{A}^{-T})e^{\mathcal{A}^Tt} \mathcal{Q}e^{\mathcal{A}t}(\psi_0+\mathcal{A}^{-1}\beta) dt\nonumber\\
            &\quad \quad + \int_{0}^{\infty} 2\beta^T\mathcal{A}^{-T}\mathcal{Q} e^{\mathcal{A}t} (\psi_0+\mathcal{A}^{-1}\beta)dt \nonumber\\
            &\quad \quad + \beta^T\mathcal{A}^{-T}\mathcal{Q}\mathcal{A}^{-1}\beta  \intoninf\nonumber\\
            &=\left(\psi_0^T +\beta^T\mathcal{A}^{-T}\right) \int_{0}^{\infty} e^{\mathcal{A}^Tt} \mathcal{Q}e^{\mathcal{A}t} dt \left(\psi_0+\mathcal{A}^{-1}\beta\right)\\
            &\quad \quad + 2\beta^T\mathcal{A}^{-T}\mathcal{Q} \int_{0}^{\infty} e^{\mathcal{A}t} dt\left(\psi_0+\mathcal{A}^{-1}\beta\right) \\
            &\quad \quad+ \beta^T\mathcal{A}^{-T}\mathcal{Q}\mathcal{A}^{-1}\beta \intoninf
        \end{align}
    \end{subequations}
    Recalling the preliminaries of the Lyapunov equation in Section \ref{subsection_preliminary}, we have 
    \begin{align}
        \label{eq_int_e_At_Q_e_At}
        \int_0^{\infty} e^{\mathcal{A}^Tt}\mathcal{Q}e^{\mathcal{A}t}dt=\mathcal{P}
    \end{align}
    Since $\mathcal{A}$ is Hurwitz, the following result holds.
    \begin{align}
    \label{eq_int_e_At}
        \int_{0}^{\infty}e^{\mathcal{A}t}dt=\mathcal{A}^{-1}e^{\mathcal{A}t}|_0^{\infty}=-\mathcal{A}^{-1}
    \end{align}
    Substituting \eqref{eq_int_e_At_Q_e_At} and \eqref{eq_int_e_At} into the result for $\int_{0}^{\infty} \psi^T\mathcal{Q}\psi dt$, the first assertion holds.

    By substituting $\psi(t)$ in \eqref{eq_psi_t} into $\int_0^{\infty} \alpha^T \psi dt$, we can easily obtain the second result in the lemma.
\end{proof}

Lemma \ref{lemma_universal_performance_computation} shows that both $\mathcal{J}_{\infty}^{\mathcal{Q}}$ and $\mathcal{J}_{\infty}^{\mathcal{\alpha}}$ have unbounded part (the cost related to $\intoninf$) and bounded part (the cost unrelated to $\intoninf$). When the coefficients of $\intoninf$ are not $0$, the index becomes unbounded. Then we will use Lemma \ref{lemma_universal_performance_computation} to compute the performance of system performance indices under the overtaking optimal controller and the near-optimal controller.

\subsection{Performance under the overtaking optimal controller}
\label{subsection_optimal_performance}
Recalling the closed-loop system \eqref{eq_closed_loop_system_A_minus_BK} under the overtaking optimal controller \eqref{eq_optimal_control_trace}, it can be rewritten as
\begin{align}
    \label{eq_closed_loop_system_optiaml}
    \dot{x}=Mx-M\bar{x}
\end{align}
where $M$ is defined as $M:=A-BK$.

Substituting $u^*(t)=-K(x-\bar{x})+\bar{u}$ into the optimal performance index $J_{\infty}^*(x_0,u^*(t))$, we have
\begin{subequations}
    \label{eq_J_star}
    \begin{align}
        J_{\infty}^*&=\frac{1}{2}\int_{0}^{\infty} x^TQx+(u^*)^TRu^*dt \nonumber \\
        &=\frac{1}{2}\int_{0}^{\infty} (x^TQx+(K(x-\bar{x})-\bar{u})^TR(K(x-\bar{x})-\bar{u}))dt\nonumber\\
        &=\frac{1}{2}\int_{0}^{\infty} \left(x^T(Q+K^TRK)x-2(K\bar{x}+\bar{u})RKx \right.\nonumber\\
        & \quad\quad \quad \quad+\left. (K\bar{x}+\bar{u})^TR(K\bar{x}+\bar{u})\right)dt\nonumber\\
        \label{eq_J_star_two_order}
        &=\frac{1}{2}\int_{0}^{\infty} x^T\left(Q+K^TRK\right)x dt\\
        \label{eq_J_star_one_order}
        &\quad -\int_0^{\infty} \left(K^TR\left(K\bar{x}+\bar{u}\right)\right)^T x dt\\
        \label{eq_J_star_zero_order}
        & \quad +\frac{1}{2}\int_{0}^{\infty} (K\bar{x}+\bar{u})^TR(K\bar{x}+\bar{u})dt
    \end{align}
\end{subequations}
where \eqref{eq_J_star_two_order} can be obtained by applying $\mathcal{J}_{\infty}^{\mathcal{Q}}$ in \eqref{eq_J_Q} with $\mathcal{Q}=Q+K^TRK$, \eqref{eq_J_star_one_order} can be obtained by applying $\mathcal{J}_{\infty}^{\alpha}$ in \eqref{eq_J_alpha} with $\alpha=K^T R(K\bar{x}+\bar{u})$, and \eqref{eq_J_star_zero_order} can be rewritten as $\frac{1}{2} (K\bar{x}+\bar{u})^TR(K\bar{x}+\bar{u})\intoninf$.

To apply the result of \eqref{eq_J_Q} to \eqref{eq_J_star_two_order}, the associated Lyapunov equation $PM+M^TP+(Q+K^TRK)=\mathbf{0}$ should be solved, for which we give the following result. 
\begin{lemma}
    \label{lemma_Lyapunov_Riccati_equation}
    The Lyapunov equation $PM+M^TP+(Q+K^TRK)=\mathbf{0}$ has a unique solution, and the solution is the same as the solution $P^*$ to the Riccati equation in \eqref{eq_Riccati}.
\end{lemma}
\begin{proof}
    Recall that $K=R^{-1}B^TP^*$ and we can know 
    \begin{align*}
        &\quad P^*A+A^TP^*-P^*BR^{-1}B^TP^*+Q\\
        &= P^*(A-BK)+(A-BK)^TP^*+P^*BK+K^TB^TP^*\\
        &\quad -P^*BR^{-1}B^TP^*+Q\\
        &=P^*M+M^T P^*+P^*BR^{-1}B^TP^*+Q\\
        &=P^*M+M^T P^*+K^TRK+Q
    \end{align*}
    Since $P^*$ is the solution to the Riccati equation, we have $P^*M+M^T P^*+K^TRK+Q=\mathbf{0}$, i.e. $P^*$ is the solution to the associated Lyapunov equation. Due to the uniqueness of the Lyapunov equation, the result holds.
\end{proof}

Then we can compute $J_{\infty}^*$ in \eqref{eq_J_star} as follows
\begin{subequations}
    \begin{align}
        J_{\infty}^*&=\frac{1}{2} \left( (x_0-\bar{x})^TP^*(x_0-\bar{x})\right.\nonumber\\
        &\quad\quad\quad -2\bar{x}^T (Q+K^TRK)(A-BK)^{-1}(x_0-\bar{x})\nonumber\\
        &\quad\quad\quad +\left. \bar{x}^T (Q+K^TRK)\bar{x}\intoninf\right)\nonumber\\
        &\quad +\left((K\bar{x}+\bar{u})^T RK(A-BK)^{-1}(x_0-\bar{x})\right.\nonumber\\
        & \quad \quad\quad \left.+ (K\bar{x}+\bar{u})^TRK(-\bar{x})\intoninf\right)\nonumber\\
        &\quad + \frac{1}{2} (K\bar{x}+\bar{u})^TR(K\bar{x}+\bar{u})\intoninf\nonumber\\
        \label{eq_J_star_final_1}
        &=\frac{1}{2}(x_0-\bar{x})^T P^* (x_0-\bar{x})\\
        \label{eq_J_star_final_2}
        &\quad +(-\bar{x}^T Q+\bar{u}^T RK)(A-BK)^{-1}(x_0-\bar{x})\\
        \label{eq_J_star_final_3}
        &\quad +\frac{1}{2}(\bar{x}^TQ\bar{x}+\bar{u}^TR\bar{u})\intoninf
    \end{align}
\end{subequations}
where the first equality is obtained by applying the results of $\mathcal{J}_{\infty}^{\mathcal{Q}}$ and $\mathcal{J}_{\infty}^{\alpha}$ in Lemma \ref{lemma_universal_performance_computation} and \ref{lemma_Lyapunov_Riccati_equation}. The unbounded part of $J_{\infty}^*$ is $\frac{1}{2}(\bar{x}^TQ\bar{x}+\bar{u}^TR\bar{u})\intoninf$, where the coefficient $\frac{1}{2}(\bar{x}^TQ\bar{x}+\bar{u}^TR\bar{u})$ is the same as the optimal steady-state cost. This coincides with the fact that $(x(t),u(t))$ converges to the optimal steady-state solution $(\bar{x}, \bar{u})$ as proved in Lemma \ref{lemma_property_of_overtaking_optimal_controller}, and thus the unbound part of the integral cost $J^*_{\infty}$ increase in the speed of the optimal steady-state cost.

\subsection{Performance under the near-optimal controller}
\label{subsection_near_optimal_performance}
Lemma \ref{lemma_equivalence_mu_lambda} shows that the implementable closed-loop system \eqref{eq_closed_loop_system_mu} is equivalent to system \eqref{eq_closed_loop_system_lambda}. Therefore, we focus on the performance analysis of system \eqref{eq_closed_loop_system_lambda} in this subsection for convenience. The closed-loop system \eqref{eq_closed_loop_system_lambda} can be written as the following compact form.
\begin{align}
    \label{eq_dot_z}
    \dot{z} = \tilde{A}z+\tilde{b}
\end{align}
where $z:=\text{col}(x,y,\lambda)$,  $\tilde{A}$ is the state matrix defined in \eqref{eq_tilde_A}, and $\tilde{b}:=\text{col}(Cd,\mathbf{0}_n,K^{\lambda}Cd)\in \mathbb{R}^{3n}$.

System performance index $J_{\infty}^p$ under the proposed controller \eqref{eq_controller_lambda} can be rewritten as follows.
    \begin{align}
        J_{\infty}^p&:= \frac{1}{2}\int_{0}^{\infty} (x^TQx+u^TRu)dt\nonumber\\
        &= \frac{1}{2} \int_{0}^{\infty}\left(x^TQx+ \norm{\big(K(x-y)+R^{-1}B^T \lambda\big)^T}_R^2 \right)dt\nonumber\\
        &=\frac{1}{2}\int_0^{\infty}\bigg(\norm{x}_{Q+K^TRK}^2+\norm{y}_{K^TRK}^2+\norm{\lambda}^2_{BR^{-1}B^T}\nonumber\\
        &\quad\quad\quad\quad\quad - 2x^TK^TRKy+2x^TK^T B^T\lambda\nonumber\\
        &\quad\quad\quad\quad\quad-2y^TK^T B^T\lambda\bigg)dt\nonumber\\
        \label{eq_J_p}
        &= \frac{1}{2} \int_{0}^{\infty} z^T \tilde{Q} z dt
    \end{align}
where $\tilde{Q}$ is 
\begin{align*}
    \tilde{Q}&:=\begin{bmatrix}Q+K^TRK& -K^T RK &K^TB^T\\
    -K^TRK & K^TRK& -K^TB^T\\
    BK& -BK&BR^{-1}B^T \end{bmatrix}\\
    &=\begin{bmatrix} Q+K^TRK&\tilde{Q}_1\\\tilde{Q}_1^T&\tilde{Q}_2\end{bmatrix}
\end{align*}
with 
\begin{align*}
    \tilde{Q}_1&:=\begin{bmatrix}-K^TRK&K^TB^T\end{bmatrix}\\
    \tilde{Q}_2&:=\begin{bmatrix}K^TRK &-K^TB^T \\ -BK & BR^{-1}B^T \end{bmatrix}
\end{align*}

Since the performance index is $J_{\infty}^p=\frac{1}{2}\int_{0}^{\infty} z^T\tilde{Q}zdt$ and the system is $\dot{z} = \tilde{A}z+\tilde{b}$, we can utilize the result of $\mathcal{J}_{\infty}^{\mathcal{Q}}$ in \eqref{eq_J_Q} with $\mathcal{Q}=\tilde{Q}, \mathcal{A}=\tilde{A}, \beta = \tilde{b}$ to compute $J_{\infty}^p$. The difficulties of applying the result of $\mathcal{J}_{\infty}^{\mathcal{Q}}$ lie in two aspects: 1) the solving of the associated Lyapunov equation $\tilde{P}\tilde{A}+\tilde{A}^T \tilde{P}+\tilde{Q}=\mathbf{0}$, and 2) the computation of $\tilde{A}^{-1}$. We solve the two problems in the following Lemma \ref{lemma_Lyapunov_equation_for_proposed_controller} and \ref{lemma_tilde_A_inv}, respectively.

\begin{lemma}
    \label{lemma_Lyapunov_equation_for_proposed_controller}
    Suppose Assumption \ref{assumption_system_property} and \ref{assumption_hurwitz_S} hold. The Lyapunov equation $\tilde{P}\tilde{A}+\tilde{A}^T \tilde{P}+\tilde{Q}=\mathbf{0}$ has a unique positive semidefinite solution $\tilde{P}^*=\begin{bmatrix}
        P^*&\mathbf{0}\\\mathbf{0}&\tilde{P}^{S^*}
    \end{bmatrix}$, where $\tilde{P}^{S^*}$ is the unique positive semidefinite solution to the Lyapunov equation $\tilde{P}^S S+S^T\tilde{P}^S +\tilde{Q}_2=\mathbf{0}$.
\end{lemma}

\begin{proof}
    Suppose $\tilde{P}=\begin{bmatrix} \tilde{P}_1&\tilde{P}_2\\\tilde{P}_3&\tilde{P_4} \end{bmatrix}$ is the solution, where $\tilde{P}_1\in \mathbb{R}^{n\times n}, \tilde{P}_2\in \mathbb{R}^{n\times 2n},\tilde{P}_3\in \mathbb{R}^{2n\times n},\tilde{P}_4\in \mathbb{R}^{2n\times 2n}$. Then $\tilde{P}\tilde{A}+\tilde{A}^T\tilde{P}+\tilde{Q}=\mathbf{0}$ can be written as
    \begin{subequations}
        \begin{align}
            \label{eq_Laypunov_1}
            \tilde{P}_1 M+M^T \tilde{P}_1+(Q+K^TRK)&=\mathbf{0}\\
            \label{eq_Laypunov_2}
            \tilde{P}_1N+\tilde{P}_2S+M^T\tilde{P}_2+\tilde{Q}_1&=\mathbf{0}\\
            \label{eq_Laypunov_3}
            \tilde{P}_3M+N^T\tilde{P}_1+S^T\tilde{P}_3+\tilde{Q}_1^T&=\mathbf{0}\\
            \label{eq_Laypunov_4}
            \tilde{P}_3N+\tilde{P}_4S+N^T\tilde{P}_2+S^T\tilde{P}_4+\tilde{Q}_2&=\mathbf{0}
        \end{align}
    \end{subequations}

    By Lemma \ref{lemma_Lyapunov_Riccati_equation} and \eqref{eq_Laypunov_1}, we can know that the solution to $\tilde{P}_1$ is $P^*$. Further, notice that when substituting $\tilde{P}_1=P^*$ into \eqref{eq_Laypunov_2}, $\tilde{P}_2=\mathbf{0}$ is naturally a solution since the following equations hold.
    \begin{align*}
        \tilde{P}_1N+\tilde{Q}_1&=P^*\begin{bmatrix}BK&-BR^{-1}B^T\end{bmatrix}+\begin{bmatrix}-K^TRK&K^TB^T\end{bmatrix}\\
        &=\begin{bmatrix}P^*BK&-K^TB^T\end{bmatrix}+\begin{bmatrix}-K^TRK&K^TB^T\end{bmatrix}\\
        &=\begin{bmatrix}
            \mathbf{0}_{n\times n}&\mathbf{0}_{n\times n}
        \end{bmatrix}
    \end{align*}
    where the definition of $K=R^{-1}B^TP^*$ is used. Similarly, $\tilde{P}_3=\mathbf{0}$ also satisfies \eqref{eq_Laypunov_3}.

    The above observation about $\tilde{P}_2,\tilde{P}_3$ motivates us to assume the solution to $\tilde{P}\tilde{A}+\tilde{A}^T\tilde{P}+\tilde{Q}=\mathbf{0}$ has the block diagonal form. As long as $\tilde{P}S+S^T\tilde{P}+\tilde{Q}_2=\mathbf{0}$ has a unique solution, the proof will be completed due to the uniqueness of the solution to $\tilde{P}\tilde{A}+\tilde{A}^T\tilde{P}+\tilde{Q}=\mathbf{0}$.

    By Lemma \ref{lemma_Hurwitz_property_T_S}, $S$ is Hurwitz. Therefore, the Lyapunov equation $\tilde{P}^SS+S^T\tilde{P}^S+\tilde{Q}_2=\mathbf{0}$ has a unique solution $P^{S^*}$, which shows that $\tilde{P}^*$ is a solution to the Lyapunov equation $\tilde{P}\tilde{A}+\tilde{A}^T \tilde{P}+\tilde{Q}=\mathbf{0}$. Therefore, by the uniqueness of the solution, the assertion in this Lemma holds.
\end{proof}

\begin{lemma}
    \label{lemma_tilde_A_inv}
    The inverse of $\tilde{A}$ is given by
    \begin{subequations}
     \label{eq_tilde_A_inv}
     \begin{align}
            \label{eq_tilde_A_inv_1}
                \tilde{A}^{-1}&=\begin{bmatrix}
                    M^{-1} & -M^{-1}NS^{-1}\\
                    0&S^{-1}
                \end{bmatrix}\\
                 \label{eq_tilde_A_inv_2}
                &=\begin{bmatrix}
                    M^{-1} & -M^{-1}NS^{-1}\\
                    0&T^{-1} (K^S)^{-1}\end{bmatrix}
            \end{align}   
    \end{subequations}
    where $(K^S)^{-1}$ and $T^{-1}$ can be represented as 
    \begin{align*}
        \left(K^S\right)^{-1}=\begin{bmatrix}
        (K^{y})^{-1}&\mathbf{0}\\\mathbf{0}&(K^{\lambda})^{-1}
    \end{bmatrix}, T^{-1}=\begin{bmatrix}
        T_1&T_2\\T_3&T_4
    \end{bmatrix}
    \end{align*} 
    with $T_1, T_2, T_3, T_4\in \mathbb{R}^{n\times n}$ satisfying the following conditions.
     \begin{subequations}
        \label{eq_T_matrix_inv}
        \begin{align}
            \label{eq_T_matrix_inv_1}
            -QT_1-A^TT_3&=I_{n}\\
            \label{eq_T_matrix_inv_2}
            -QT_2-A^TT_4&=\mathbf{0}\\
            \label{eq_T_matrix_inv_3}
            AT_1-BR^{-1}B^TT_3&=\mathbf{0}\\
            \label{eq_T_matrix_inv_4}
            AT_2-BR^{-1}B^TT_4&=I_n
        \end{align}
    \end{subequations}
\end{lemma}
\begin{proof}
    The mathematical form of $\tilde{A}^{-1}$ in \eqref{eq_tilde_A_inv} can be verified by comparing the result of $\tilde{A}*\tilde{A}^{-1}$ with $I_{3n}$. Conditions in \eqref{eq_T_matrix_inv} can be proved by substituting $T$ in \eqref{eq_T_def} into $T*T^{-1}=I_{2n}$.
\end{proof}
    
   

Lemma \ref{lemma_Lyapunov_equation_for_proposed_controller} and Lemma \ref{lemma_tilde_A_inv} give the simplified results for the solution to the Lyapunov equation $\tilde{P}\tilde{A}+\tilde{A}^T\tilde{P}+\tilde{Q}=\mathbf{0}$ and the inverse of $\tilde{A}$. We are ready to obtain the system performance index $J_{\infty}^p$ under the proposed near-optimal controller in the following theorem. 

\begin{theorem}
    The performance index $J_{\infty}^p=\frac{1}{2}\int_{0}^{\infty} z^T\tilde{Q}zdt$ for  system \eqref{eq_dot_z} is given below
    \begin{align}\label{perf_near}
        J_{\infty}^p=J_{\infty}^*+\frac{1}{2}\norm{\begin{bmatrix}
            y_0-\bar{x} \\ \lambda_0-\bar{\lambda}
        \end{bmatrix}}^2_{\tilde{P}^{S^*}}
    \end{align}
    where $J_{\infty}^*$ is the performance index for the overtaking optimal controller, $\tilde{P}^{S^*}$ is the solution to Lyapunov equation $\tilde{P}^S S+S^T\tilde{P}^S +\tilde{Q}_2=0$ in Lemma \ref{lemma_Lyapunov_equation_for_proposed_controller}, and $y_0,\lambda_0$ are the initial value for $y$ and $\lambda$.
\end{theorem}
\begin{proof}
    Since system \eqref{eq_dot_z} is an asymptotically stable linear system with $z$ as its state and $J_{\infty}^p$ is an integral over the quadratic function about $z$, we can apply the result of $\mathcal{J}_{\infty}^{\mathcal{Q}}$ in \eqref{eq_J_Q} with $\mathcal{Q}=\tilde{Q}$ to compute $J_{\infty}^p$. Recalling \eqref{General_JQ}, the result consists of two parts: one related to $\intoninf$ and the other unrelated to $\intoninf$.
    
    First, we compute the part related to $\intoninf$ in $J_{\infty}^p$, which is $\frac{1}{2}\norm{\tilde{A}^{-1}\tilde{b}}_{\tilde{Q}}\intoninf$ according to \eqref{eq_J_Q}. The value of $\tilde{A}^{-1}$ is 
    \begin{align*}
        \tilde{A}^{-1}
        &=\begin{bmatrix}
            M^{-1}&-M^{-1}NS^{-1}\\\mathbf{0}_{2n\times n}&S^{-1}
        \end{bmatrix}\\
        &=\begin{bmatrix}
            M^{-1}&M^{-1}NT^{-1}\\
            \mathbf{0}&T^{-1}
        \end{bmatrix} \begin{bmatrix}I_n&\mathbf{0}\\
        \mathbf{0}&(K^S)^{-1}\end{bmatrix}
    \end{align*}
    The value of $M^{-1}NT^{-1}$ is computed  as follows.
    \begin{align}
        &\quad M^{-1}NT^{-1} \nonumber\\
        &=M^{-1}\begin{bmatrix}BK&-BR^{-1}B^T\end{bmatrix}\begin{bmatrix}
            T_1&T_2\\T_3&T_4
        \end{bmatrix} \nonumber\\
        &=M^{-1} \begin{bmatrix}
            BKT_1-BR^{-1}B^T T_3&BKT_2-BR^{-1}B^T T_4
        \end{bmatrix} \nonumber\\
        &=M^{-1}\begin{bmatrix}
            BKT_1-A T_1&BKT_2-AT_2+I_{n}
        \end{bmatrix} \nonumber\\
        \label{eq_M_1_N_T_1}
        &=M^{-1}\begin{bmatrix}
            -(A-BK)T_1 &I_{n}-(A-BK)T_2
        \end{bmatrix}
    \end{align}
    where in the third equation, $BR^{-1}B^TT_3=AT_1$ in \eqref{eq_T_matrix_inv_3} and $BR^{-1}B^TT_4=AT_2-I_{n}$ in \eqref{eq_T_matrix_inv_4} are used. Since $M$ is defined as $M=A-BK$, we can know $M^{-1}NT^{-1}=\begin{bmatrix}-T_1&M^{-1}-T_2\end{bmatrix}$. Therefore, $\tilde{A}^{-1}\tilde{b}$ can be written as 
    \begin{align*}
        &\quad \tilde{A}^{-1}\tilde{b}\\
        &=\begin{bmatrix}
            M^{-1}&T_1&T_2-M^{-1}\\
            \mathbf{0}&T_1&T_2\\
            \mathbf{0}&T_3&T_4
        \end{bmatrix} \begin{bmatrix}I_n&\mathbf{0}\\
        \mathbf{0}&(K^S)^{-1}\end{bmatrix}\begin{bmatrix}
            Cd\\\mathbf{0}\\K^{\lambda}Cd
        \end{bmatrix}\\
        &=\begin{bmatrix}
            M^{-1}&T_1&T_2-M^{-1}\\
            \mathbf{0}&T_1&T_2\\
            \mathbf{0}&T_3&T_4
        \end{bmatrix} \begin{bmatrix}
            Cd\\\mathbf{0}\\Cd
        \end{bmatrix}\\
        &=\begin{bmatrix}
            T_2Cd\\T_2Cd\\T_4Cd
        \end{bmatrix}
    \end{align*}
    Eliminating $\bar{u}$ in \eqref{eq_KKT_condition_static_OP_2} by $\bar{u}=R^{-1}B^T \bar{\lambda}$, \eqref{eq_KKT_condition_static_OP} is equivalent to 
    \begin{align*}
        Q\bar{x}+A^T\bar{\lambda}&=\mathbf{0}\\
        A\bar{x}-BR^{-1}B^T\bar{\lambda}+Cd&=\mathbf{0}
    \end{align*}
    which can be written as $T\begin{bmatrix}
            \bar{x}\\\bar{\lambda}
        \end{bmatrix}=\begin{bmatrix}
            \mathbf{0}\\-Cd
        \end{bmatrix}$. Therefore, we have 
        \begin{align*}
            \begin{bmatrix}
            \bar{x}\\\bar{\lambda}
        \end{bmatrix}=T^{-1}\begin{bmatrix}
            \mathbf{0}\\-Cd
        \end{bmatrix}=\begin{bmatrix}
            T_1&T_2\\T_3&T_4
        \end{bmatrix}\begin{bmatrix}
            \mathbf{0}\\-Cd
        \end{bmatrix}=\begin{bmatrix}
            -T_2Cd\\-T_4Cd
        \end{bmatrix}
        \end{align*}from which we can know
        \begin{align*}
            \tilde{A}^{-1}\tilde{b}=\begin{bmatrix}
                T_2Cd\\T_2Cd\\T_4Cd
            \end{bmatrix}=-\begin{bmatrix}
                \bar{x}\\\bar{x}\\\bar{\lambda}
            \end{bmatrix}
        \end{align*}
        Then the part related to $\intoninf$ in the performance index is 
        \begin{align*}
            &\quad \frac{1}{2}\norm{\tilde{A}^{-1}\tilde{b}}_{\tilde{Q}}^2\intoninf \\
            &=\frac{1}{2}\begin{bmatrix}
                \bar{x}\\\bar{x}\\\bar{\lambda}
            \end{bmatrix}^T\begin{bmatrix}Q+K^TRK& -K^T RK &K^TB^T\\
                -K^TRK & K^TRK& -K^TB^T\\
                BK& -BK&BR^{-1}B^T \end{bmatrix}\begin{bmatrix}
                    \bar{x}\\\bar{x}\\\bar{\lambda}
                \end{bmatrix}\intoninf\\
                &= \frac{1}{2}\left(\bar{x}^TQ\bar{x}+\bar{\lambda}^TBR^{-1}B^T\bar{\lambda}\right)\intoninf\\
                &=\frac{1}{2}\left(\bar{x}^TQ\bar{x}+\bar{u}^TR\bar{u}\right)\intoninf
        \end{align*}
    which is the same as the part related to $\intoninf$ of $J_{\infty}^*$ in \eqref{eq_J_star_final_3}.

    For the part unrelated to $\intoninf$ in $J_{\infty}^p$, it can be computed by applying the result of $\mathcal{J}_{\infty}^{\mathcal{Q}}$ in \eqref{eq_J_Q} as follows. 
    \begin{align*}
        & \quad \frac{1}{2} (z_0+\tilde{A}^{-1}\tilde{b})^T \tilde{P}^* (z_0+\tilde{A}^{-1}\tilde{b}) + \tilde{b}^T \tilde{A}^{-T} 
        \tilde{Q} \tilde{A}^{-1} (z_0+\tilde{A}^{-1}\tilde{b})\\
        &=\frac{1}{2}\!\begin{bmatrix}
            x_0\!-\!\bar{x}\\y_0\!-\!\bar{x}\\\lambda_0\!-\!\bar{\lambda}
        \end{bmatrix}^T\! \begin{bmatrix}
            P^*&\mathbf{0}\\\mathbf{0}&\tilde{P}^{S^*}
        \end{bmatrix}\! \begin{bmatrix}
            x_0\!-\!\bar{x}\\y_0\!-\!\bar{x}\\\lambda_0\!-\!\bar{\lambda}
        \end{bmatrix}\!-\!\begin{bmatrix}
            \bar{x}\\\bar{x}\\\bar{\lambda}
        \end{bmatrix}^T\!\tilde{Q}\!\tilde{A}^{-1}\!\begin{bmatrix}
            x_0-\bar{x}\\y_0-\bar{x}\\\lambda_0-\bar{\lambda}
        \end{bmatrix}\\
        &=\frac{1}{2}\!\norm{x_0\!-\!\bar{x}}^2_{P^*}\!+\!\frac{1}{2}\norm{\begin{bmatrix}
            y_0-\bar{x} \\ \lambda_0-\bar{\lambda}
        \end{bmatrix}}^2_{\tilde{P}^{S^*}}\!-
        \!\begin{bmatrix}
            \bar{x}\\\bar{x}\\\bar{\lambda}
        \end{bmatrix}^T \! \tilde{Q}\!\tilde{A}^{-1}\!\begin{bmatrix}
            x_0\!-\!\bar{x}\\y_0\!-\!\bar{x}\\\lambda_0\!-\!\bar{\lambda}
        \end{bmatrix}
    \end{align*}
    Comparing the result with the part unrelated to $\intoninf$ in $J_{\infty}^*$ in \eqref{eq_J_star_final_1} and \eqref{eq_J_star_final_2}, it is sufficient to prove 
    \begin{align}
        \label{eq_performance_difference_of_bounded_part_target}
        \!\!-\!\!\begin{bmatrix}
            \bar{x}\\\bar{x}\\\bar{\lambda}
        \end{bmatrix}^{\!T\!}\!\!\tilde{Q}\tilde{A}^{-1}\!\!\begin{bmatrix}
            x_0\!-\!\bar{x}\\y_0\!-\!\bar{x}\\\lambda_0\!-\!\bar{\lambda}
        \end{bmatrix}\!=(-\bar{x}^T \!Q\!+\!\bar{u}^T \!RK)M^{\!-\!1}(x_0\!-\!\bar{x})
    \end{align}
    
    The lefthand side yields 
    \begin{align*}
        &\quad \!-\!\!\!\begin{bmatrix}
            \bar{x}\\\bar{x}\\\bar{\lambda}
        \end{bmatrix}^{\!T\!}\!\! \tilde{Q}\tilde{A}^{-1}\begin{bmatrix}
            x_0-\bar{x}\\y_0-\bar{x}\\\lambda_0-\bar{\lambda}
        \end{bmatrix}\\
        &=-\begin{bmatrix}
            Q\bar{x}+K^TB^T\bar{\lambda}\\-K^TB^T\bar{\lambda}\\BR^{-1}B^T\bar{\lambda}
        \end{bmatrix}^T \tilde{A}^{-1}\begin{bmatrix}
            x_0-\bar{x}\\y_0-\bar{x}\\\lambda_0-\bar{\lambda}
        \end{bmatrix}
    \end{align*}
    We separate $\tilde{A}^{-1}$ into two parts as follows.
    \begin{align*}
        \tilde{A}^{-1}&=\begin{bmatrix}
        M^{-1}&\mathbf{0}\\\mathbf{0}&\mathbf{0}
        \end{bmatrix}+\begin{bmatrix}
        \mathbf{0}&M^{-1}NS^{-1}\\\mathbf{0}&S^{-1}
        \end{bmatrix}
    \end{align*}
    Then we have
    \begin{subequations}
        \begin{align}
            &\quad -\begin{bmatrix}
                \bar{x}\\\bar{x}\\\bar{\lambda}
            \end{bmatrix}^T\tilde{Q}\tilde{A}^{-1}\begin{bmatrix}
                x_0-\bar{x}\\y_0-\bar{x}\\\lambda_0-\bar{\lambda}
            \end{bmatrix}\nonumber\\
            \label{eq_performance_final_J_p_bounded_1}
            &= -(\bar{x}^TQ+\bar{\lambda}^TBK) M^{-1} (x_0-\bar{x})\\
            \label{eq_performance_final_J_p_bounded_2}
            &\quad -\!\begin{bmatrix}
                Q\bar{x}+K^{T}B^T\bar{\lambda}\\-K^TB^T\bar{\lambda}\\BR^{-1}B^T\bar{\lambda}
            \end{bmatrix}^{T} \! \begin{bmatrix}
                \mathbf{0}&M^{-1}NS^{-1}\\ \mathbf{0} &S^{-1}
            \end{bmatrix} \! \begin{bmatrix}
                x_0-\bar{x}\\y_0-\bar{x}\\\lambda_0-\bar{\lambda}
            \end{bmatrix}
        \end{align}
    \end{subequations}
    
    Since \eqref{eq_performance_final_J_p_bounded_1} is equal to the righthand side of \eqref{eq_performance_difference_of_bounded_part_target}, we only need to show \eqref{eq_performance_final_J_p_bounded_2} is $0$. 
    
    By $S^{-1} = T^{-1} (K^S)^{-1}$, we have
    \begin{align*}
        \begin{bmatrix}
                \mathbf{0}&M^{-1}NS^{-1}\\ \mathbf{0} &S^{-1}
            \end{bmatrix} =  \begin{bmatrix}
                \mathbf{0}&M^{-1}NT^{-1}\\ \mathbf{0} &T^{-1}
            \end{bmatrix} \begin{bmatrix}
                I_{n} &\mathbf{0} \\ \mathbf{0} &(K^S)^{-1}
            \end{bmatrix} 
    \end{align*}
    Therefore, it is sufficient to show 
    \begin{align*}
    \begin{bmatrix}
                 Q\bar{x}+K^TB^T\bar{\lambda}\\-K^TB^T\bar{\lambda}\\BR^{-1}B^T\bar{\lambda}
        \end{bmatrix}^T  \begin{bmatrix}
                \mathbf{0}&M^{-1}NT^{-1}\\ \mathbf{0} &T^{-1}
            \end{bmatrix}=\mathbf{0}    
    \end{align*}
    By \eqref{eq_M_1_N_T_1}, we have 
    \begin{align}
       &\!\quad \begin{bmatrix}
                 Q\bar{x}+K^TB^T\bar{\lambda}\\-K^TB^T\bar{\lambda}\\BR^{-1}B^T\bar{\lambda}
        \end{bmatrix}^T  \begin{bmatrix}
                \mathbf{0}&M^{-1}NT^{-1}\\ \mathbf{0} &T^{-1}
            \end{bmatrix}\nonumber\\
        &\!=\! \!\begin{bmatrix}
            Q\bar{x}+K^TB^T\bar{\lambda}\\-K^TB^T\bar{\lambda}\\BR^{-1}B^T\bar{\lambda}
        \end{bmatrix}^T \begin{bmatrix}
            \mathbf{0}&T_1&T_2-M^{-1}\\
            \mathbf{0}&T_1&T_2\\
            \mathbf{0}&T_3&T_4
        \end{bmatrix}\nonumber\\
        \label{eq_performance_equal_simlifying_0}
        &\!=\!\!\begin{bmatrix}
            \mathbf{0}\\
            T_1^TQ\bar{x}+T_3^TBR^{-1}B^T\bar{\lambda}\\
            T_2^TQ\bar{x}\!-\!M^{-T}(Q\bar{x}\!+\!K^TB^T\bar{\lambda})\!+\!T_4^TBR^{-1}B^T\bar{\lambda}
        \end{bmatrix}^T
    \end{align}
    
    Substituting $\bar{x} = -T_2Cd, \bar{\lambda} = -T_4 Cd$ into the second row in \eqref{eq_performance_equal_simlifying_0}, we have
        \begin{align}
            &\quad T_1^TQ\bar{x}+T_3^TBR^{-1}B^T\bar{\lambda} \nonumber\\
            &=  -(T_1^T Q T_2 +T_3^T BR^{-1}B^T T_4)Cd\nonumber\\
            \label{eq_performance_equal_simlifying_1}
            & =-\left(T_1^T\left(-A^T T_4\right)+\left(AT_1\right)^TT_4\right)Cd\\
            &=\mathbf{0}\nonumber
        \end{align}
    where \eqref{eq_performance_equal_simlifying_1} holds by substituting the results in \eqref{eq_T_matrix_inv_2} and \eqref{eq_T_matrix_inv_3}.
    
    Similarly, substituting $\bar{x} = -T_2Cd, \bar{\lambda} = -T_4 Cd$ into the third row in \eqref{eq_performance_equal_simlifying_0}, we have
    \begin{subequations}
     \begin{align}
        &\quad T_2^TQ\bar{x}-M^{-T}(Q\bar{x}+K^TB^T\bar{\lambda})+T_4^TBR^{-1}B^T\bar{\lambda}\nonumber \\
        &=\Big(-T_2^TQT_2+M^{-T}\big(QT_2+K^TB^TT_4\big) \nonumber\\
        &\quad\qquad - T_4^TBR^{-1}B^TT_4\Big)Cd\nonumber\\ 
        &=\Big(T_2^TA^TT_4+M^{-1}(-A^TT_4+K^TB^TT_4)\nonumber\\
        \label{eq_performance_equal_simlifying_2_2}
        & \quad\qquad-(AT_2-I_n)^T T_4\Big)Cd\\
        &=\Big(T_2A^TT_4-M^{-1}(A-BK)^TT_4\nonumber\\
        &\quad \qquad-T_2^TA^TT_4+T_4\Big)Cd\nonumber\\
        \label{eq_performance_equal_simlifying_2_3}
        &=(T_2A^TT_4-T_4-T_2^TA^TT_4+T_4)Cd\\
        &=\mathbf{0}\nonumber
    \end{align}
    \end{subequations}
    where \eqref{eq_performance_equal_simlifying_2_2} is obtained by substituting \eqref{eq_T_matrix_inv_2} and \eqref{eq_T_matrix_inv_4}, \eqref{eq_performance_equal_simlifying_2_3} holds due the the definition of $M=A-BK$.

    To conclude, we prove \eqref{eq_performance_difference_of_bounded_part_target} holds since \eqref{eq_performance_final_J_p_bounded_1} is the same to the righthand side of \eqref{eq_performance_difference_of_bounded_part_target} and \eqref{eq_performance_final_J_p_bounded_2} is $0$. Therefore, $J_{\infty}^p=J_{\infty}^*+\frac{1}{2}\norm{\begin{bmatrix}
        y_0-\bar{x} \\ \lambda_0-\bar{\lambda}
    \end{bmatrix}}^2_{\tilde{P}^{S^*}}$ holds. 
    
    This completes the proof.
\end{proof}
\begin{remark}
    Compared with the optimal value, the performance gap of the proposed controller is $\frac{1}{2}\norm{\begin{bmatrix}
        y_0-\bar{x} \\ \lambda_0-\bar{\lambda}
    \end{bmatrix}}^2_{\tilde{P}^{S^*}}$, which is related to the differences between the initial state $(y_0, \lambda_0)$ and the optimal steady state $(\bar{x}, \bar{\lambda})$. Since $\tilde{P}^{S^*}$ is positive semidefinite, system performance under the overtaking optimal controller is always no worse than that under the near-optimal controller, which is consistent with the definition of overtaking optimality. On the other hand, the positive semidefiniteness of $\tilde{P}^{S^*}$ means that $J_{\infty}^p$ can be equal to $J_{\infty}^*$ sometimes, which will be formally stated in the following subsection.
\end{remark}

\subsection{Discussion}
\label{subsection_discussion}
In this subsection, we first give the condition that eliminates the transient performance gap between the overtaking optimal controller and the near-optimal controller. Then, the relationship between the performance gap and the control parameters is discussed.

The following lemma gives the conditions for eliminating the performance gap.
\begin{lemma}
    \label{corollary_zero_gap}
    $J_{\infty}^p-J_{\infty}^* = 0$ holds if and only if $(y_0-\bar{y}, \lambda_0-\bar{\lambda})$ is in the unobservable subspace of $(S, \tilde{C})$ with $\tilde{C}:= \begin{bmatrix}RK& -B^T\end{bmatrix}$.
\end{lemma}
\begin{proof}
    Denote by $\alpha = \text{col}(y_0-\bar{y}, \lambda_0-\bar{\lambda})$. $J_{\infty}^p-J_{\infty}^* = 0$ holds if and only if $\alpha^T \tilde{P}^{S^*} \alpha=0$. Since $\tilde{P}^{S^*}$ is the solution to the Lyapunov equation $\tilde{P}^S S+S^T \tilde{P}^S+\tilde{Q}_2=0$, it can be expressed as follows.
    \begin{align*}
        \tilde{P}^{S^*}&=\int_{0}^{\infty} e^{S^T t}\tilde{Q}_2 e^{St}dt\\
        &=\int_{0}^{\infty} e^{S^T t} \tilde{C}^TR^{-1}\tilde{C}  e^{St}dt
    \end{align*}
    where the fact $\tilde{C}^TR^{-1}\tilde{C} = \tilde{Q}_2$ is used. Then we have 
    \begin{align*}
        \alpha^T \tilde{P}^{S^*}\alpha&=\int_{0}^{\infty} \alpha e^{S^T t} \tilde{C}^TR^{-1}\tilde{C}  e^{St}\alpha dt\\
        &=\int_{0}^{\infty} \norm{\tilde{C}e^{St} \alpha}_{R^{-1}}^2 dt
    \end{align*} 
    Therefore, $J_{\infty}^p-J_{\infty}^*=0$ if and only if $\tilde{C}e^{St} \alpha=\mathbf{0}$. Recalling the definition of unobservable subspace in \eqref{eq_unobservable_subspace_definition}, we can know $\tilde{C}e^{St}\alpha=\mathbf{0}$ implies $\alpha$ is in the unobservable subspace of $(S, \tilde{C})$. This completes the proof.
\end{proof}

Lemma \ref{corollary_zero_gap} shows that the proposed controller can achieve \textbf{P2} for some initial states of $y$ and $\lambda$ when $(S, \tilde{C})$ is not observable. We give the following example to show that the conditions can be satisfied by some systems.
\begin{example}
    Consider the following system.
    \begin{align}
        \dot{x} &= \begin{bmatrix}
            0&1\\0&0
        \end{bmatrix}x+\begin{bmatrix}
            0\\1
        \end{bmatrix}u+\begin{bmatrix}
            0\\1
        \end{bmatrix}\\
        Q & = \begin{bmatrix}
            1&0\\0&1
        \end{bmatrix}, R = 1, K^y = K^{\lambda} = I_2
    \end{align}

    We first show the example satisfies Assumption \ref{assumption_system_property} and \ref{assumption_hurwitz_S}. The system above is with the controllable canonical form, and thus stabilizable. The observability matrix of $(A, Q)$ is of full column rank since $Q$ is of full column rank. Then $(A, Q)$ is observable, and thus detectable. Therefore, Assumption \ref{assumption_system_property} is satisfied. Since $Q$ is positive definite, Assumption \ref{assumption_hurwitz_S} also holds by Lemma \ref{lemma_Hurwitz_property_T_S}. 
    
    Then, it can be easily obtained that $\bar{x} = (0, 0)^T, \bar{\lambda} = (0, 1)^T, K = (1, \sqrt {3})$. Then we have 
    \begin{align*}
        S = \begin{bmatrix}
            -1&0&0&0\\0&-1 &-1&0\\0&1&0&0\\0&0&0&-1
        \end{bmatrix},\tilde{C}=\begin{bmatrix}
            1&\sqrt{3}&0&-1
        \end{bmatrix}
    \end{align*}
    The unobservable space of $(S, \tilde{C})$ is $\{x\in \mathbb{R}^4| x=k\alpha, \forall k \in \mathbb{R}\}$ with $\alpha = (1,0,0,1)$. Therefore, when $y_0=(1, 0)^T$ and $\lambda_0 = (0, 2)^T$, $J^p_{\infty}-J^*_{\infty}$ will be zero.
\end{example}

The following lemma shows the relationship between the performance gap and the control gain $K^S$.

\begin{lemma}
    \label{corollary_1_k}
    Under the same initial value $(y_0, \lambda_0)$, the performance gap $J_{\infty}^p-J_{\infty}^*$ is inversely proportional to the gain parameter $K^S$, i.e. if $K^{S_1}=k K^{S_2}$, the corresponding performance difference will be 
    \begin{align*}
        J_{\infty}^{p_1}-J_{\infty}^*=\frac{1}{k}(J_{\infty}^{p_2}-J_{\infty}^*)
    \end{align*}
    where $J_{\infty}^{p_1}$ and $J_{\infty}^{p_2}$ are the system performance under control gains $K^{S_1}$ and $K^{S_2}$, respectively. 
\end{lemma}
\begin{proof}
    Suppose $\tilde{P}^S$ is the solution to the Lyapunov equation $\tilde{P}^S S+S^T\tilde{P}^S+\tilde{Q}_2=\mathbf{0}$, which yields
    \begin{align}
    \label{eq_lyapunov_K_S_T}
        \tilde{P}^S(K^S T)+(K^ST)^T\tilde{P}^S+\tilde{Q}_2=\mathbf{0}
    \end{align}
    Suppose $\tilde{P}_1$ is the solution to the Lyapunov equation when $K^S=K^{S_1}$, i.e. 
    \begin{align*}
        \tilde{P}_1(K^{S_1} T)+(K^{S_1} T)^T\tilde{P}_1+\tilde{Q}_2=\mathbf{0}
    \end{align*}
    We can know that $\tilde{P}_2={k}\tilde{P}_1$ will be the solution to \eqref{eq_lyapunov_K_S_T} when $K^S=K^{S_2}$ since $K^{S_1}=kK^{S_2}$. Then, we have 
    \begin{align*}
        J_{\infty}^{p_1}-J^*&=\frac{1}{2}\norm{\begin{bmatrix}
            y_0-\bar{x} \\ \lambda_0-\bar{\lambda}
        \end{bmatrix}}^2_{\tilde{P}_1}\\
        &=\frac{1}{k}\cdot \frac{1}{2}\norm{\begin{bmatrix}
            y_0-\bar{x} \\ \lambda_0-\bar{\lambda}
        \end{bmatrix}}^2_{\tilde{P}_2}\\
        &=\frac{1}{k}(J_{\infty}^{p_2}-J_{\infty}^*)
    \end{align*}
which shows that the performance gap is inversely proportional to the control gain. 
\end{proof}

Lemma \ref{corollary_1_k} shows that \textbf{P2} can be realized with any given error bound by increasing the control gain. Together with Theorem \ref{theorem_asymptotic_stablity}, both \textbf{P1} and \textbf{P2} are achieved simultaneously by the near-optimal controller.

\section{Case study}
\label{section_case_study}
\begin{figure}[t]
    \centering
    \includegraphics[width=0.15\textwidth]{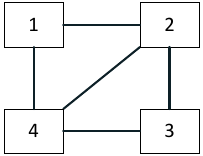}
    \caption{Structure of the four-bus system} 
    \label{fig_structure_of_four_bus} 
    \vspace{-0.3cm}
\end{figure}

\begin{table}[t]
    \footnotesize
    \centering
    \caption{System parameters}
    \label{label_system_parameters}
    \begin{subtable}[t]{0.5\textwidth}
    \centering
        \caption{Bus Parameters}
    \label{table_bus_parameter}
	\begin{tabular}{c|cccc}
        \hline
        Bus & Bus 1&Bus 2 &Bus 3 &Bus 4\\
        \hline
        $M_i$&2.0 &1.5&1.8&3.0 \\
        $D_i$&2.0 &2.0&3.0&4.0\\
        $q_i$&15.0 &10.0&12.0&18.0\\
        $r_i$&1.0 &1.0&2.0&1.5\\
        \hline
	\end{tabular}
    \end{subtable}
    \\ 
    \hspace{0.1mm}
     \begin{subtable}[t]{0.5\textwidth}
    \centering
        \caption{Line parameters}
    \label{table_line_parameter}
	\begin{tabular}{c|ccccc}
		\hline
		Line $(i,j)$& $(1,2)$& $(2,3)$ & $(3,4)$&$(4,1)$&$(2,4)$ \\
        \hline
        Reactance $B_{ij}$&1.5&1.0&2.0&1.8&2.5\\
		\hline
	\end{tabular}
    \end{subtable}
\end{table}
\begin{table}[t]
\footnotesize
\centering
        \caption{Optimal steady state}
    \label{table_optimal_steady_state}
    \begin{tabular}{c|ccc}
        \hline
        Bus&$\bar{\delta}_{i}$&$\bar{\omega}_i$&$\bar{u}_i$ \\
        \hline
        {Bus 1} & -0.201 &-0.298&1.491\\
        {Bus 2} & 0.500 &-0.298&1.491\\
        {Bus 3} & 0.713 &-0.298&0.745\\
        {Bus 4} & 0 &-0.298&0.994\\
        \hline
	\end{tabular}
\end{table}

\begin{figure}[t]
    \centering
    \includegraphics[width=0.4\textwidth]{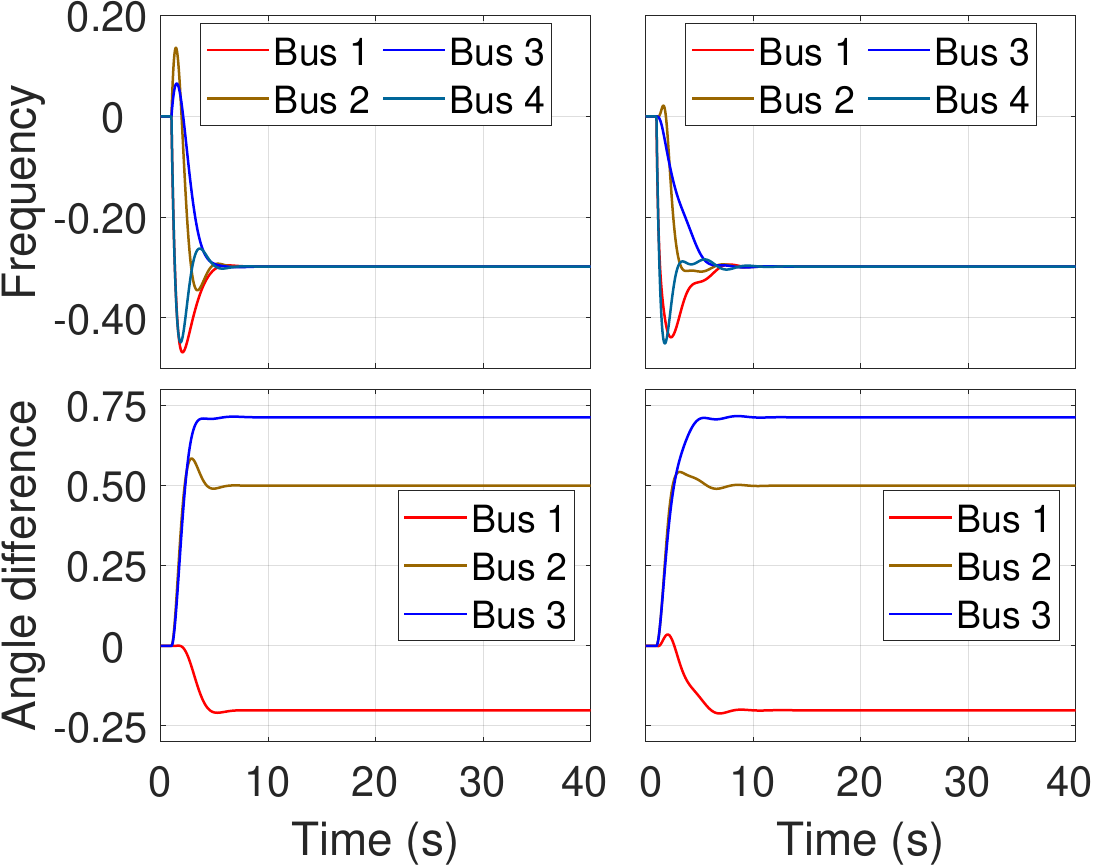}
    \caption{System frequencies and angle differences under optimal controller(Left) and proposed controller(Right)}
    \label{fig_state_trajectory_frequency_angle}
\end{figure}

\subsection{Simulation settings of a power system}
We perform the simulations on a simplified power system model with 4 buses, whose structure is shown in Fig. \ref{fig_structure_of_four_bus}. Every bus has a phase angle $\theta_i$ and a frequency $\omega_i$. The frequency control dynamics of the power system is as follows.
\begin{subequations}
    \label{eq_power_dynamics}
    \begin{align}
        \dot{\theta}_i&=\omega_i,\forall i\in \mathcal{N}\\
        M_i\dot{\omega}_i&=u_i-d_i-D_i\omega_i-\sum_{j\in \mathcal{N}_i} B_{ij}(\theta_i-\theta_j),\forall i\in\mathcal{N}
    \end{align}
\end{subequations}
where $u_i$ is the controllable power, $d_i$ is the power disturbance, $M_i>0$ is the inertia of bus $i$, $D_i>0$ is the dampness, $B_{ij}>0$ is the reactance of the line between bus $i$ and $j$, $\mathcal{N}_i$ is the set of the buses connected to $i$, $\mathcal{N}$ is the bus set, i.e. $\mathcal{N}=\{1,2,3,4\}$.

We choose to minimize the following target
\begin{align}
    f = \frac{1}{2}\sum_{i\in\mathcal{N}}(q_i \omega_i^2 + r_iu_i^2)
\end{align}
where $q_i$ is the weight for the cost of frequency $\omega_i$, and $r_i$ is the coefficient for the cost of power $u_i$. 

It can be seen that $\theta_i$ has a redundant state since the dynamics \eqref{eq_power_dynamics} is related to the angle difference $\theta_i-\theta_j$. We eliminate $\theta_4$ by introducing new variables $\delta_i,i\in\{1,2,3\}$ to represent the angle differences, i.e. $\delta_i = \theta_i-\theta_4$. Specifically, we define $\delta_4 \equiv 0$. Then \eqref{eq_power_dynamics} can be written as 
\begin{subequations}
    \label{eq_power_dynamics_delta}
    \begin{align}
        \dot{\delta}_i&=\omega_i-\omega_4,\forall i\in \{1,2,3\}\\
        M_i\dot{\omega}_i&=u_i-d_i-D_i\omega_i-\sum_{j\in \mathcal{N}_i} B_{ij}(\delta_i-\delta_j),\forall i\in \mathcal{N}
    \end{align}
\end{subequations}

The system is asymptotically stable under a constant $u_i$ \cite{cuiTPS2023reinforcement}. Therefore, the stabilizable and detectable property in Assumption \ref{assumption_system_property} is satisfied. The asymptotic stability of primal-dual dynamics in Assumption \ref{assumption_hurwitz_S} has been widely investigated in frequency control in power systems, e.g. \cite{zhao2014design, chenTCNS2020distributed}, and we omit the proof here.

The system parameters are given in Table \ref{label_system_parameters}.

\subsection{Simulation results}
The system works at the equilibrium when $t<1s$. At $t_0=1s$, we apply a power disturbance $d$ in the system, where $d=\text{col}(d_1,d_2,d_3,d_4)=\text{col}(3.5, 0, 0, 4.5)$.

\begin{figure}[t]
    \centering
    \includegraphics[width=0.35\textwidth]{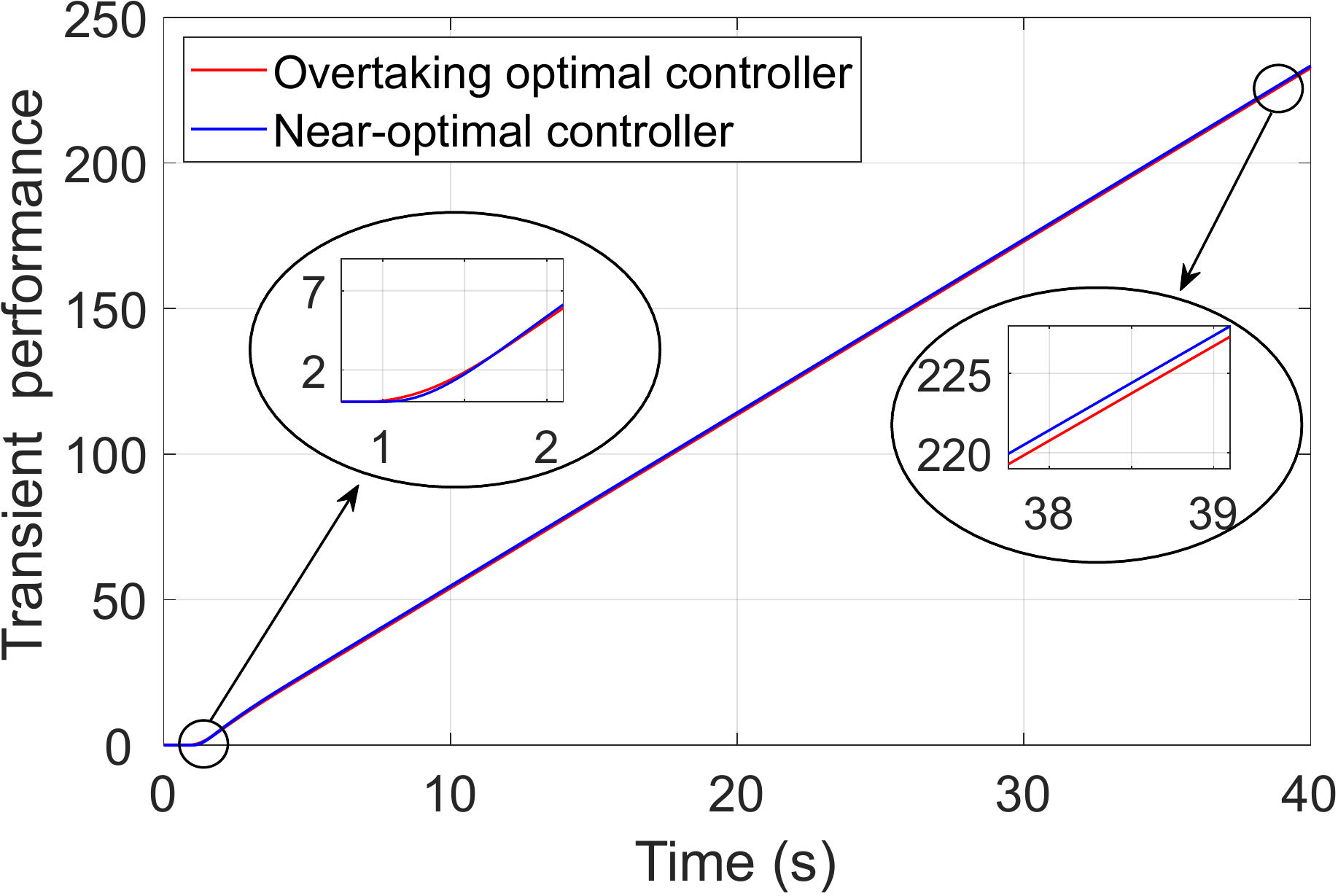}
    \caption{Transient Performance index $J_T$}
    \label{fig_performance_J_T}
\end{figure}

\begin{figure}[t]
    \centering
    \includegraphics[width=0.35\textwidth]{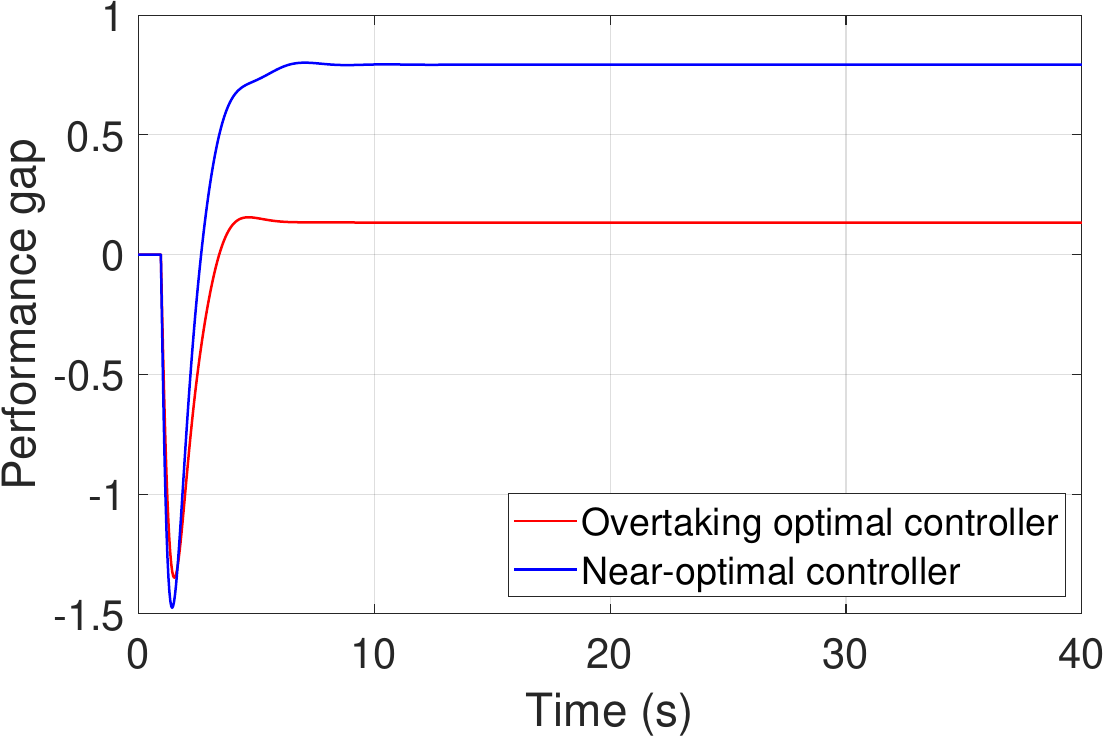}
    \caption{Transient performance gap with optimal steady state $J_T-\bar{f}\cdot(T-t_0)$}
    \label{fig_performance_J_T_difference_J_bar}
\end{figure}

\begin{figure}[t]
    \centering
    \includegraphics[width=0.35\textwidth]{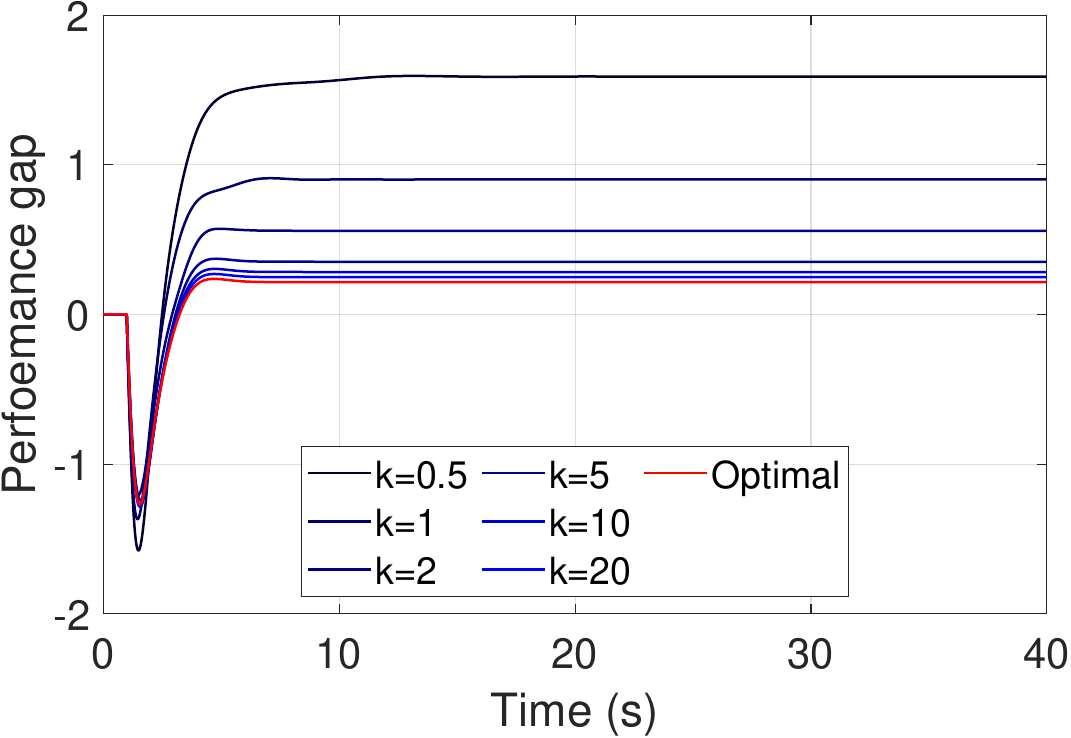}
    \caption{Transient performance gap $J_T-\bar{f}(T-t_0)$}
    \label{fig_performance_difference}
\end{figure}
\begin{figure}[t]
    \centering
    \includegraphics[width=0.35\textwidth]{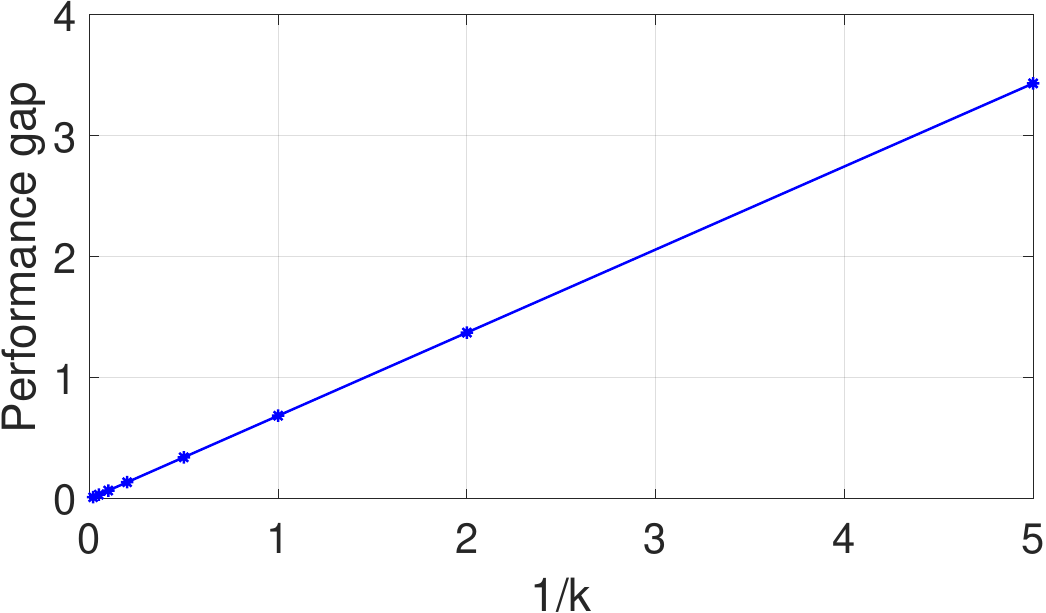}
    \caption{Transient performance gap $J_{T}^{p}-J_{T}^*$ with $T=40s$ under different $k$} 
    \label{fig_performance_under_k} 
\end{figure}

We compare the results under the overtaking optimal controller with the proposed near-optimal controller when introducing the disturbance. First, the optimal steady state $(\bar{\omega}_i, \bar{\delta}_i, \bar{u}_i),i\in\{1,2,3,4\}$ is computed by quadratic programming and shown in Table \ref{table_optimal_steady_state}. The optimal control feedback gain $K$ is 
\begin{align*}
    \begin{bmatrix}
        -0.1  & 0.09& -0.03&  2.31 & 0.  & -0.01& -0.03\\
        -0.09 & 0.06&  0.04&  0.   & 1.76&  0.01&  0.03\\
        -0.01 & 0.01&  0.02& -0.01 & 0.01&  0.88&  0.  \\
        -0.06 & 0.04&  0.03& -0.01 & 0.01&  0.  &  1.29\\
    \end{bmatrix}
\end{align*}
where the first three columns are gains for $\delta_i,i\in\{1,2,3\}$, and the others are gains for $\omega_i,i\in \{1,2,3,4\}$.

 In the proposed near-optimal controller, we first set $K^y=\bar{K}^y$ with $\bar{K}^y$ as the diagonal matrix with elements $(0.2, 0.2 ,0.2, 0.5, 0.5, 0.5, 0.5)$ on the diagonal, and $K^{\lambda}=\bar{K}^{\lambda}$ with $\bar{K}^{\lambda}$ as the diagonal matrix with elements $(10, 10, 10, 10, 10, 10, 10)$ on the diagonal. The trajectories of $\delta_i$ and $ \omega_i$ are shown in Fig. \ref{fig_state_trajectory_frequency_angle}. Compared with optimal state $\bar{\omega}_i, \bar{\delta}_i$, both the overtaking optimal controller and the near-optimal controller steer the system to the optimal steady state. 

To illustrate the transient gap of the performance index under different controllers, we plot the finite-time horizon performance index $J_T$ in Fig. \ref{fig_performance_J_T}, which shows that 1) during the transient process, the performance indices of both controllers increase after the disturbance occurs, 2) when the system enters the steady state (e.g. $t>10$s), both performance indices increase linearly as time goes on, and 3) the performance index of the proposed controller is a little higher than the optimal controller. Since the performance indices increase at the rate of $\bar{f}=\frac{1}{2}\sum_{i=1}^4(q_i \bar{\omega}_i^2 + r_i\bar{u}_i^2)$ after the system enters the steady state, we plot the gap of the performances with $\bar{f}\cdot (T-t_0)$ in Fig. \ref{fig_performance_J_T_difference_J_bar}, which clearly shows the differences of the proposed controller and the optimal controller.

We then perform the simulations under different control gains $K^y, K^{\lambda}$. We set $K^y=k\bar{K}^y,K^{\lambda}=k\bar{K}^{\lambda}$ with $k=0.5,1,2,5,10, 20$, respectively. The trajectories of $J_T-\bar{f}\cdot (T-t_0)$ under optimal controller and proposed controller with different $k$ are shown in Fig. \ref{fig_performance_difference}. As we increase $k$, the performance gaps between the proposed and optimal control are closer. To further validate the inversely proportional relationship between the performance gaps $J_{\infty}^p-J_{\infty}^*$ and parameters $K^y, K^{\lambda}$, we plot the values of $J_T^p-J_T^*$ at $T=40$s to approximate $J_\infty^p-J_\infty^*$ in Fig. \ref{fig_performance_under_k}. The result is consistent with the theoretical analysis.

\section{Conclusion}

We have investigated the optimal control problem in both the steady state and the transient process of an LTI system with unknown disturbances. The overtaking optimal controller can achieve the optimal performance in the two stages, which, however, requires detailed information of the disturbance. Leveraging the primal-dual dynamics, a near-optimal controller is proposed, which can steer the system with unknown disturbances to the optimal steady state. The transient performance gap compared with the optimal value is proved to be zero or inversely proportional to the control gains, depending on the initial state. 
Since the proposed near-optimal controller is independent of the disturbances, it can also be implemented in time-varying systems. The immediate future work is the performance analysis with time-varying disturbances. Besides, the near-optimal controller is currently implemented in a centralized manner. How to extend it to a distributed way is also significant in many systems, such as power systems and communication networks.

\appendices
\makeatletter
\@addtoreset{equation}{section}
\@addtoreset{theorem}{section}
\makeatother
\renewcommand{\theequation}{A.\arabic{equation}}
\renewcommand{\thetheorem}{A.\arabic{theorem}}
\setcounter{equation}{0}
\section*{Appendix}

\subsubsection*{Proof of Theorem \ref{theorem_overtaking_optimal}} 
\label{appendix_theorem_overtaking_optimal}
The optimal control problem is written as follows.
\begin{align*}
    \min_{u} &\quad J_{\infty}(x_0,u)=\int_0^{\infty} f(x,u)dt\\
    s.t. &\qquad   \dot{x}=g(x,u)
\end{align*}
where $f(x,u)$ is defined in \eqref{eq_steady_optimization_problem_target}, and $g(x,u):=Ax+Bu+Cd$. We will prove $u^*(t)$ in \eqref{eq_optimal_control_trace} to be the overtaking optimal controller. 

First, we give the solution to the Hamilton-Jacobi-Bellman equation in the following lemma.



\begin{lemma}
    For the Hamilton-Jacobi-Bellman equation in \eqref{eq_HJB}, 
    the function $V(x, t)$ in \eqref{eq_V_t} is the solution.
    \begin{align}
        \label{eq_HJB}
        \frac{\partial{V}}{\partial t}+f\left(x,h\left(x,\frac{\partial V}{\partial x}\right)\right)+\frac{\partial V}{\partial x}  ~g\left(x, h\left(x,\frac{\partial V}{\partial x}\right)\right)=0
    \end{align}
    \begin{align}
        \label{eq_V_t}
        V(x, t)&=-\tilde{L}(\bar{x}, \bar{\lambda})t+\frac{1}{2}(x-\bar{x})^TP^*(x-\bar{x})+\bar{\lambda}^T(x-\bar{x})
    \end{align}
    where $h(\cdot, \cdot)$ and  is the function defined in \eqref{eq_h_definition}, $\tilde{L}(\cdot, \cdot)$ is the function defined in \eqref{eq_reduced_Lagrangian}, $\bar{x}, \bar{\lambda}$ are defined in Lemma \ref{lemma_steady_state_problem}, and $P^*$ is the solution to the Riccati equation in Theorem \ref{theorem_overtaking_optimal}.
\end{lemma}
\begin{proof}
    Define $\mathbb{L}(x, u, t)$ as 
    \begin{align*}
        \mathbb{L}(x, u, t) &:= f(x, u)+\frac{\partial V}{\partial t}+\left(\frac{\partial V}{\partial x}\right)^T g(x,u)
    \end{align*}
    Then, the Hamilton-Jacobi-Bellman equation can be written as $\mathbb{L}(x, h(x, \frac{\partial V}{\partial x}), t)=0$. Thus, we first simplify $\mathbb{L}(x, u, t)$ as follows.

    \begin{align*}
        \mathbb{L}(x, u, t)&\overset{\textcircled{1}}{=}\frac{1}{2}(x^TQx+u^TRu)-\tilde{L}(\bar{x}, \bar{\lambda})\\
        &\quad+(P^*(x-\bar{x})+\bar{\lambda})^T(Ax+Bu+Cd)\\
        &\overset{\textcircled{2}}{=}\frac{1}{2}x^TQx+\frac{1}{2}u^TRu-\frac{1}{2} \bar{x}^TQ\bar{x}+\frac{1}{2}\bar{\lambda}BR^{-1}B^T\bar{\lambda}\\
        &\quad +(x-\bar{x})^TP^*A(x-\bar{x})+\bar{\lambda}^TBu\\
        &\quad +(x-\bar{x})^TP^*(A\bar{x}+Bu+Cd)+\bar{\lambda}^TA(x-\bar{x})\\
        &\overset{\textcircled{3}}{=}\frac{1}{2} (x-\bar{x})^T(Q+2P^*A)(x-\bar{x})+\bar{x}^TQ(x-\bar{x})\\
        &\quad +\frac{1}{2}u^TRu+\frac{1}{2}\bar{u}^TR\bar{u}+\bar{\lambda}^TBu+\bar{\lambda}^TA(x-\bar{x})\\
        &\quad +(x-\bar{x})^TP^*B(u-\bar{u})\\
        &\overset{\textcircled{4}}{=}\frac{1}{2}(x-\bar{x})^TP^*BR^{-1}B^TP^*(x-\bar{x})\\
        &\quad +\frac{1}{2}(u-\bar{u})^TR(u-\bar{u})+(Q\bar{x}+A^T\bar{\lambda})^T(x-\bar{x})\\
        &\quad + \bar{\lambda}^TBu+\bar{u}^TRu+(x-\bar{x})^TP^*B(u-\bar{u})\\
        &\overset{\textcircled{5}}{=}\frac{1}{2}(x-\bar{x})^TK^TRK(x-\bar{x})\\
        &\quad +\frac{1}{2}(u-\bar{u})^TR(u-\bar{u})+(x-\bar{x})^TRK^T(u-\bar{u})\\
        &\overset{\textcircled{6}}{=}\frac{1}{2}\norm{u-\bar{u}+K(x-\bar{x})}^2_R
    \end{align*}
    where step \textcircled{1} is obtained by substituting $\frac{\partial V}{\partial t}$ and $\frac{\partial V}{\partial x}$; in step \textcircled{2}, the definition of $\tilde{L}$ in \eqref{eq_reduced_Lagrangian} is used; step \textcircled{3} is obtained by collecting the terms containing $x-\bar{x}$ and substituting $\bar{u}=-R^{-1}B^T\bar{\lambda}$; in step \textcircled{4}, the fact $P^*$ is the solution to the Riccati equation is used; step \textcircled{5} utilizes the definition of $K=R^{-1}B^TP^*$ and the condition $Q\bar{x}+A^T\bar{\lambda}=\mathbf{0}$ in \eqref{eq_KKT_condition_static_OP_1}; the result in step \textcircled{6} can be directly verified by computation. 

    Since $h(x, \frac{\partial V}{\partial x}) = -R^{-1}B^T(P^*(x-\bar{x})+\bar{\lambda})=-K(x-\bar{x})+\bar{u}$, we have
    \begin{align*}
        \mathbb{L}(x, h(x, \frac{\partial V}{\partial x}), t)=\frac{1}{2}\norm{h(x, \frac{\partial V}{\partial x})-\bar{u}+K(x-\bar{x})}_R^2=0
    \end{align*}
    which implies $V(x,t)$ in \eqref{eq_V_t} is the solution to the Hamilton-Jacobi-Bellman equation.
\end{proof}

Then we are ready to prove that $u^*=-K(x(t)-\bar{x})+\bar{u}$ is the overtaking optimal control. Denote by $x^*(t)$ the trajectory under $u^*(t)$. For any other control $u'(t)$ and the corresponding trajectory $x'(t)$, we have
\begin{align*}
    &\quad \int_{0}^{T} f(x^*, u^*)dt-\int_{0}^T f(x', u')dt\\
    &=\int_{0}^{T}\mathbb{L}(x^*, u^*, t)dt-\int_{0}^T\mathbb{L}(x',u',t)dt\\
    &\quad -\int_{0}^T \frac{d}{dt}(V(x^*,t))dt+\int_{0}^T\frac{d}{dt}(V(x',t))dt\\
    &=\int_{0}^{T}\mathbb{L}(x^*, u^*, t)dt-\int_{0}^T\mathbb{L}(x',u',t)dt\\
    &\quad -V(x^*(T), T)+V(x'(T), T)\\
    &=\frac{1}{2}\int_{0}^{T}\left(\norm{u^*\!-\!\bar{u}\!+\!K(x^*\!-\!\bar{x})}_R^2 \!-\!\norm{u'\!-\!\bar{u}\!+\!K(x'\!-\!\bar{x})}_R^2\right) dt\\
    &\quad-\frac{1}{2}\norm{x^*(T)-\bar{x}}_{P^*}^2+\frac{1}{2}\norm{x'(T)-\bar{x}}_{P^*}^2\!\\
    &\quad-\bar{\lambda}^T(x^*(T)-x'(T))
\end{align*}
Since $u^*(t)=K(x^*(t)-\bar{x})$ and $x^*(T)\to \bar{x}$, we have
\begin{align*}
    &\quad \limsup_{T\to \infty}\int_{0}^T(f(x^*, u^*)-f(x',u'))\\
    &=\limsup_{T\to \infty} \left(-\frac{1}{2}\int_{0}^T \norm{u'-\bar{u}+K(x'-\bar{x})}_R^2dt\right.\\
    &\quad \quad\quad \quad \quad \quad \left.+\frac{1}{2}\norm{x'(T)-\bar{x}}_{P^*}^2+\bar{\lambda}^T(x'(T)-\bar{x})\right)
\end{align*}

Define $\tilde{x}:=x'-\bar{x}$ and $\tilde{u}:=u'-\bar{u}+K(x'-\bar{x})$. We can know that 
\begin{align}
    \label{eq_dot_tilde_x}
    \dot{\tilde{x}}=(A-BK)\tilde{x}+B\tilde{u}
\end{align}
and 
\begin{align}
    &\quad \limsup_{T\to\infty} \int_0^T(f(x^*, u^*)-f(x',u'))\nonumber\\
    \label{eq_lim_sup_result}
    &=\limsup_{T\to\infty} \left(-\frac{1}{2}\int_{0}^T \norm{\tilde{u}}_R^2dt+\frac{1}{2}\norm{\tilde{x}}_{P^*}^2+\bar{\lambda}\tilde{x}\right)
\end{align}

We then prove \eqref{eq_lim_sup_result} is no larger than $0$ as follows. 
\begin{itemize}
    \item When $\int_{0}^{\infty}\norm{\tilde{u}}_R^2dt$ is bounded, i.e. $\tilde{u}\in \mathcal{L}_2$, the $\mathcal{L}_2$ stabliity theory shows that $\tilde{x}$ converges to $\mathbf{0}$ due to the Hurwitz property of $(A-BK)$. Thus, \eqref{eq_lim_sup_result} is equal to $-\frac{1}{2}\int_{0}^{\infty}\norm{\tilde{u}}_R^2dt$, which is nonpositive.
    \item When $\int_{0}^{\infty}\norm{\tilde{u}}_R^2dt$ is unbounded, since $u'(t)$ is an admissible control, $x'(t)$ and $u'(t)$ are bounded. Therefore, $\tilde{x}(t)$ is bounded. Then we have $-\frac{1}{2}\int_{0}^{T}\norm{\tilde{u}}_R^2dt +\frac{1}{2} \norm{\tilde{x}(T)}_{P^*}^2+\bar{\lambda}\tilde{x}(T)\to -\infty$ as $T\to \infty$.
\end{itemize}

Therefore, we conclude that $u^*(t)$ is the overtaking optimal control.

\section*{References}

\bibliographystyle{IEEEtran}
\bibliography{mybib}

\begin{IEEEbiography}[{\includegraphics[width=1in,height=1.25in,clip,keepaspectratio]{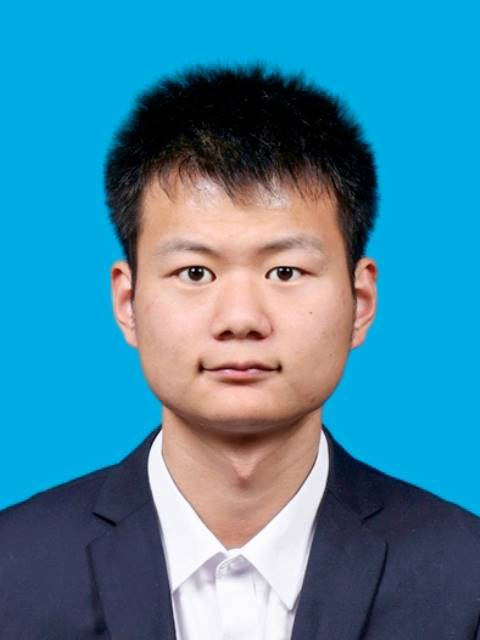}}]{Ming Li} received the B.Sc. degree in automation from Tsinghua University, Beijing, China, in 2018. He is currently pursuing the Ph.D degree in the Department of Automation, Shanghai Jiao Tong University, Shanghai, China. His research interests include game theory, frequency control in power system, and optimal control.
\end{IEEEbiography}

\begin{IEEEbiography}[{\includegraphics[width=1in,height=1.25in,clip,keepaspectratio]{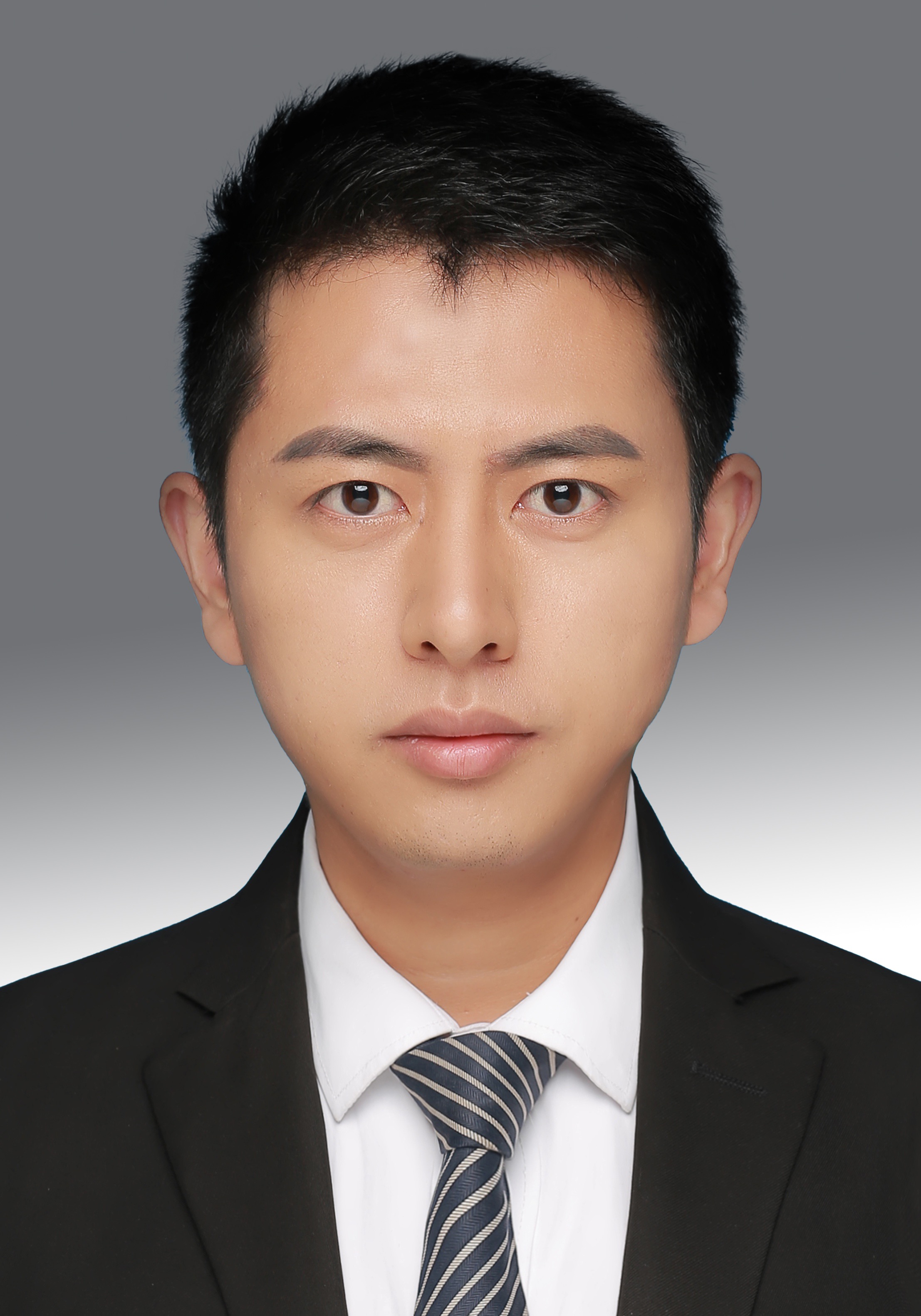}}]{Zhaojian Wang}
	(M'18)  received the B.Sc. degree in electrical engineering from Tianjin University, Tianjin, China, in 2013, and the Ph.D. degree in electrical engineering from Tsinghua University, Beijing, China, in 2018. From 2016 to 2017, he was a Visiting Ph.D. student with the California Institute of Technology, Pasadena, CA, USA. From 2018 to 2020, he was a Postdoctoral Scholar with Tsinghua University. He is currently an Associate Professor with Shanghai Jiao Tong University, Shanghai, China. His research interests include stability analysis, optimal control, V2G regulation and game theory based decision-making in energy and power systems. He is an Associate Editor for IEEE Transactions on Power Systems and IEEE Systems Journal. He was one of the top 5 reviewers of the IEEE Transactions on Smart Grid (TSG) for 2023 and the Best Paper award of  IEEE Transactions on Power Systems (2018-2020).
\end{IEEEbiography}

\begin{IEEEbiography}[{\includegraphics[width=1in,height=1.25in,clip,keepaspectratio]{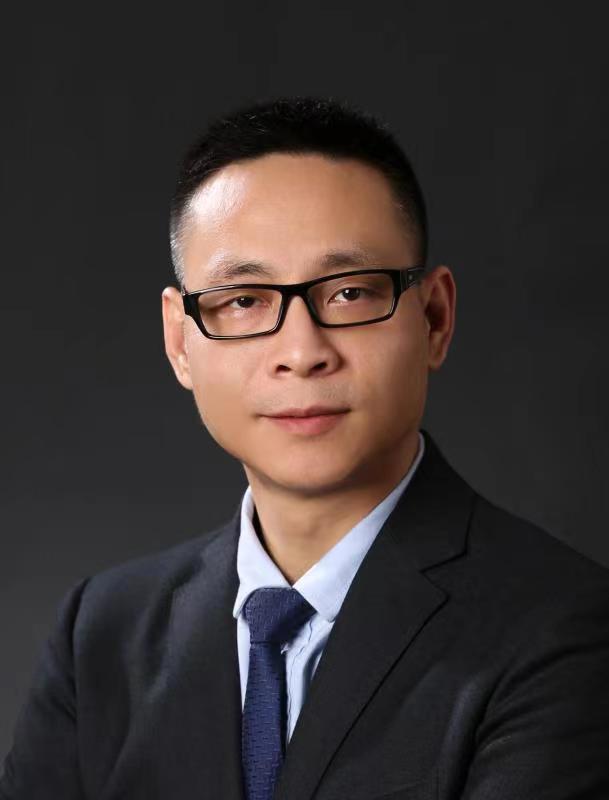}}]{Feng Liu}
	(M'10-SM'18) received the B.Sc. and Ph.D. degrees in electrical engineering from Tsinghua University, Beijing, China, in 1999 and 2004, respectively. He is currently an Associate Professor at Tsinghua University. From 2015 to 2016, he was a visiting associate at the California Institute of Technology, CA, USA. Dr. Feng Liu’s research interests include stability analysis, optimal control, robust dispatch, and game theory-based decision-making in energy and power systems. He is the author/coauthor of more than 300 peer-reviewed technical papers and four books and holds more than 20 issued/pending patents. Dr. Liu is an IET Fellow. He is an associated editor of several international journals, including IEEE Transactions on Power Systems, IEEE Transactions on Smart Grid, and Control Engineering Practice. He also served as a guest editor of IEEE Transactions on Energy Conversion. Dr. Feng Liu was the winner of the outstanding (Top 3) AE award of  IEEE Transactions on Smart Grid (2023), the Best Paper award of  IEEE Transactions on Power Systems (2018-2020) and several international conferences.
\end{IEEEbiography}
\begin{IEEEbiography}[{\includegraphics[width=1in,height=1.25in,clip,keepaspectratio]{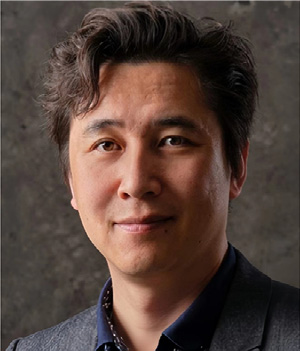}}]{Ming Cao} (Fellow, IEEE) received the bachelor’s degree in 1999 and the master’s degree in 2002 from Tsinghua University, China, and the Ph.D. degree in 2007 from Yale University, New Haven, CT, USA, all in electrical engineering. He is currently a Professor of systems and control with the Engineering and Technology Institute (ENTEG), the University of Groningen, Groningen, the Netherlands, where he started as an Assistant Professor in 2008. From 2007 to 2008, he was a Postdoctoral Research Associate with the Department of Mechanical and Aerospace Engineering, Princeton University, Princeton, NJ, USA. He worked as a Research Intern in 2006 with the Mathematical Sciences Department, IBM T. J. Watson Research Center, USA. His main research interest is in autonomous agents and multi-agent systems, decision making dynamics and complex networks.
\end{IEEEbiography}

\begin{IEEEbiography}[{\includegraphics[width=1in,height=1.25in,clip,keepaspectratio]{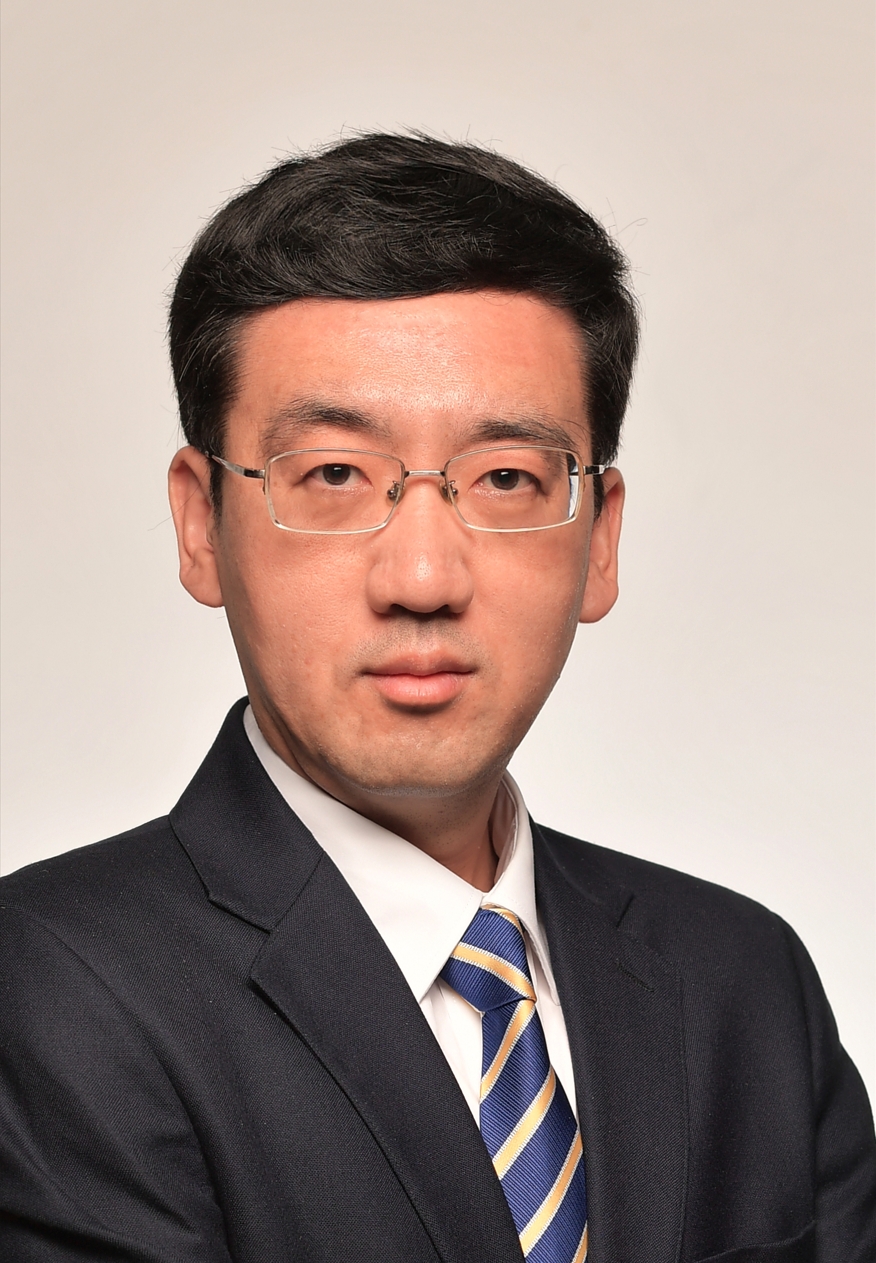}}]{Bo Yang} received the PhD degree in electrical engineering from the City University of Hong Kong, Hong Kong, in 2009. He is currently a full professor with Shanghai Jiao Tong University, Shanghai, China. He held visiting positions at KTH, Sweden; New York University, USA. His research interests include optimization and control for energy networks and Internet of Things. He has been the principal investigator in several research projects, including the NSFC Key Project and the National Science Fund for Distinguished Young Scholars, China. He was a recipient of the Ministry of Education Natural Science Award, the Shanghai Technological Invention Award, the Shanghai Rising Star Program.
\end{IEEEbiography}

\end{document}